\documentclass[12pt]{article}
\usepackage[english]{babel}
\usepackage[utf8]{inputenc}
\usepackage[T1]{fontenc}
\usepackage{listings}
\usepackage{amsbsy}
\usepackage{amsfonts}
\usepackage{amsmath}
\usepackage{amssymb}
\usepackage{amsthm}
\usepackage{graphicx}
\usepackage[pdftex]{color}
\usepackage{dsfont}
\usepackage{enumerate}
\usepackage{float}
\usepackage{geometry, calc, color, setspace}%
\usepackage{indentfirst}%Para indentar os parágrafos automáticamente
\usepackage{longtable}
\usepackage{multirow}
\usepackage{natbib}
\usepackage{xr-hyper}
\usepackage{hyperref,bookmark}
\newcommand{\bs}{\boldsymbol}

\newtheorem{theorem}{Theorem}[section]

\newtheorem{lemma}[theorem]{Lemma}
\newtheorem{proposition}[theorem]{Proposition}
\newtheorem{corollary}[theorem]{Corollary}
\newtheorem{remark}[theorem]{Remark}

\hypersetup{
  colorlinks=true,
  linkcolor=blue,
  citecolor=red,
  filecolor=blue,
  urlcolor=blue,
}
\newcommand{\blind}{1}

\protect

\begin{document}

\def\spacingset#1{\renewcommand{\baselinestretch}%
{#1}\small\normalsize} \spacingset{1}

%\spacingset{1.45} % DON'T change the spacing!
%%%%%%%%%%%%%%%%%%%%%%%%%%%%%%%%%%%%%%%%%%%%%%%%%%%%%%%%%%%%%%%%%%%%%%%%%%%%%%

\if1\blind
{
  \title{\bf Bessel regression model: Robustness to analyze bounded data}
  \author{Wagner Barreto-Souza$^\star$\footnote{Corresponding Author. E-mail: wagner.barretosouza@kaust.edu.sa},\,\, Vin\'icius D. Mayrink$^\star$\footnote{E-mail: vdm@est.ufmg.br}\,\, and\, Alexandre B. Simas$^\#$\footnote{E-mail: alexandre@mat.ufpb.br}\vspace{.3cm}\\  	
    {\small\it $^\star$Departamento de Estat\'istica, Universidade Federal de Minas Gerais, Belo Horizonte, Brazil}\\
             {\small\it $^\#$Departamento de Matem\'atica, Universidade Federal da Para\'iba, Jo\~ao Pessoa, Brazil}\\
         \small\it $^\ast$Statistics Program, King Abdullah University of Science and Technology, Thuwal, Saudi Arabia}
  \maketitle
} \fi
\if0\blind
{
  \bigskip
  \bigskip
  \bigskip
  \begin{center}
    {\LARGE\bf Bessel regression model: Robustness to analyze bounded data}
\end{center}
  \medskip
} \fi

\bigskip
\addtocontents{toc}{\protect\setcounter{tocdepth}{1}}

\begin{abstract}

Beta regression has been extensively used by statisticians and practitioners to model bounded continuous data and there is no strong and similar competitor having its main features. A class of normalized inverse-Gaussian (N-IG) process was introduced in the literature, being explored in the Bayesian context as a powerful alternative to the Dirichlet process. Until this moment, no attention has been paid for the univariate N-IG distribution in the classical inference. In this paper, we propose the bessel regression based on the univariate N-IG distribution, which is a robust alternative to the beta model. This robustness is illustrated through simulated and real data applications. The estimation of the parameters is done through an Expectation-Maximization algorithm and the paper discusses how to perform inference. A useful and practical discrimination procedure is proposed for model selection between bessel and beta regressions. Monte Carlo simulation results are presented to verify the finite-sample behavior of the EM-based estimators and the discrimination procedure. Further, the performances of the regressions are evaluated under misspecification, which is a critical point showing the robustness of the proposed model. Finally, three empirical illustrations are explored to confront results from bessel and beta regressions.

\end{abstract}
Misspecification; Model selection.

\vfill
%\newpage

\section{Introduction}
\label{intro}

The beta distribution is one of the most common distributions used in real-life to handle continuous bounded data. It is well-known that if $Z$ follows a beta distribution with parameters $\alpha>0$ and $\beta>0$, then it satisfies the stochastic representation 
\begin{eqnarray}\label{storep}
Z\stackrel{d}{=}\dfrac{Y_1}{Y_1+Y_2},
\end{eqnarray}
where $Y_1$ and $Y_2$ are independent gamma random variables with scale parameter equal to 1 and shape parameters $\alpha$ and $\beta$, respectively. The multivariate extension of the beta distribution, known as Dirichlet distribution, appeared in \cite{fer73} and since then many studies have emerged using this multivariate option as a Bayesian nonparametric approach. 

The beta regression model was first introduced by \cite{fercri04} based on a mean-precision parameterization of the beta distribution. Many papers have arisen from this model to deal with bias corrections and non-linear extensions; for instance, see \cite{sv2006} and \cite{sbsr10}. Diagnostic tools and residuals for the beta regression were considered in \cite{efc08a}, \cite{efc08b}, \cite{rs11}, \cite{fec11}, \cite{chi11}, \cite{chi13}, \cite{asb14}, \cite{esc17} and \cite{mmpc18}.

More recently, \cite{bss17} proposed a full EM algorithm approach for the beta regression model including estimation, inference, diagnostic tools and residuals. This approach has advantages over the direct maximization of the likelihood function, mainly related to the estimation of parameters associated with the precision term.

Inference and variable selection for the beta regression is considered in \cite{cns12}, \cite{zzll14}, \cite{bcn15} and \cite{bcn17}. All these references assume independence among the response variables. By relaxing this assumption and accounting for a dependence among the responses the authors in \cite{rcn09}, \cite{gv14}, \cite{ffc15}, \cite{bccn18} and \cite{petal19} proposed time series models based on the beta distribution.

Beta regression has been extensively used by researchers to model bounded continuous data in areas such as medicine, pharmacology, odontology, education and political science; see, for example, the slides of Prof. Silvia Ferrari presented in February 2013 at the $13^{th}$ Brazilian School of Regression Models (available online at \url{https://www.ime.usp.br/~sferrari/13EMRslidesSilvia.pdf}). In the current literature, there is no strong and robust competitor to the beta regression capable of handling continuous bounded data having its main features such as: (a) stochastic representation as given in (\ref{storep}); (b) mean-precision parameterization; (c) EM-algorithm for parameter estimation. As an example, the Kumaraswamy distribution \citep{k80} emerged as a possible alternative to the beta distribution; for more details on this, see the paper by \cite{j09}, where comparisons between these distributions are provided. However, this option does not have a simple formulation for the mean. In other words, building a regression model based on the mean of the Kumaraswamy is cumbersome and, for this reason, this topic has never been explored in the literature. A median-dispersion Kumaraswamy regression was proposed in \cite{mb13}. As far as we known, the Kumaraswamy distribution does not have a stochastic representation such as (\ref{storep}). We highlight that stochastic representations are important, since they may justify some models arising naturally in certain real situations. Moreover, this allows us to obtain an EM-algorithm for estimating the parameters. Another alternative is to transform (non-linearly) the response variable to be $\mathbb R$-valued and then use, for example, a normal linear model; the logit transformation is a popular choice. We emphasize that it is not clear which non-linear transformation is adequate in practical situations. Furthermore, with this approach, data are analyzed on a non-original scale, which complicates the interpretation of the parameters in applied studies.

An important alternative to the beta regression is the simplex model by \cite{barjor91}, which is a special case from the exponential dispersion models \citep{jor92}. A recent class of Johnson $S_B$ regression models for analyzing univariate bounded data was proposed by \cite{lembaz2016}, which is obtained by transformation of a symmetric continuous random variable with support on $\mathbb R$ and with a regression structure considered for the median. Although these models have their own merits, they do not share the same features of the beta model, which is the focus of the present paper.

With this in mind, it should be natural that a strong competitor of the beta regression can be developed based on a distribution satisfying a stochastic representation in the form (\ref{storep}). \cite{lijetal05} proposed an alternative to the Dirichlet process, named normalized inverse-Gaussian (N-IG) process, which is based on ratios of inverse-Gaussian (IG) random variables. The univariate distribution of this process satisfies the stochastic representation given in (\ref{storep}) by replacing the gamma assumption by the inverse-Gaussian distribution. This process has been explored in the Bayesian context \citep{lijetal05}. On the other hand, no attention has been paid for the univariate N-IG distribution in the classical inference until now. 

In this paper, the robust model called bessel regression is proposed as an alternative to the beta regression. A central aspect is the fact that the bessel model is based on the univariate N-IG distribution, which we call bessel distribution in an analogy to the beta distribution as explained in the next section. We give emphasis for the importance of introducing this alternative, since there is no regression model for continuous bounded data having many interesting features as those in the beta case. In statistical modeling, it is in general a good strategy to attack a particular problem by using several tools, rather than trusting in a single option. One of the main motivations of the present paper is to provide another appealing model for a regression setting with bounded continuous response. The main idea is that the bessel will be considered by researchers in a joint data analysis with the beta regression. Besides this, we list below the main contributions of the paper:
\begin{itemize}
	\item The bessel regression is shown to be a robust alternative to the beta model. This was detected in both simulated and real data analyses. By robustness here, we refer to the ability of the bessel model to perform well under misspecification and in practical situations, when compared with the beta regression.
	\item Due to the stochastic representation of the bessel distribution, a full EM algorithm is obtained for the bessel regression. This allows estimation, inference and diagnostic tools. In particular, the maximum likelihood estimation can be done through this approach.
	\item A discrimination procedure is proposed in order to select between the bessel and beta models, which is extremely valuable in practical situations. The idea here can be extended for a broader model selection involving other types of regressions.
\end{itemize}

This paper is unfold in the following manner. In Section \ref{sec:besseldist}, we describe the bessel distribution and present some of its important properties to develop this work. The bessel regression model is introduced in Section \ref{sec:reg}. Further, we propose estimation of the parameters through an EM algorithm and discuss how to perform inference. A discrimination procedure to select between the bessel and beta regressions is proposed in Section \ref{sec:discrimination}. Simulation results to check the finite-sample behavior of the proposed EM-based estimators for the bessel regression is presented in Section \ref{sec:sim}. In Section \ref{sec:miss}, we evaluate the performance of the bessel and beta models under misspecification. Finally, three empirical illustrations comparing both models are investigated in Section \ref{sec:emp}. Concluding remarks are addressed in Section \ref{sec:concluding}.

\section{Bessel distribution}\label{sec:besseldist}

An alternative to the Dirichlet process was proposed by \cite{lijetal05}, named normalized inverse-Gaussian process, which is built using similar arguments as those for the Dirichlet case. More specifically, the authors considered the ratio among inverse-Gaussian random variables instead of gamma variables. In particular, the univariate case named normalized inverse-Gaussian distribution (shortly denoted by N-IG) is obtained by the stochastic representation (\ref{storep}) with $Y_1$ and $Y_2$ being independent inverse-Gaussian random variables with scale parameter equal to 1 and shape parameters $\alpha>0$ and $\beta>0$. We denote $Y\sim\mbox{IG}(\alpha)$ and the corresponding density function is written as follows
\begin{eqnarray*}
	h(y)=\dfrac{\alpha}{\sqrt{2\pi}}y^{-3/2}\exp\left\{-\dfrac{1}{2}\left(\dfrac{\alpha^2}{y}+y\right)+\alpha\right\},\quad y>0.
\end{eqnarray*}
The density of the univariate N-IG distribution is given by
\begin{eqnarray}\label{densitybessel}
f(z)=\dfrac{\alpha\beta e^{\alpha+\beta}}{\pi z(1-z)}\left(\alpha^2 z+(1-z)\beta^2\right)^{-1/2}K_1\left(\sqrt{\dfrac{\alpha^2}{1-z}+\dfrac{\beta^2}{z}}\right),\quad z\in(0,1),
\end{eqnarray}
where $K_1(\cdot)$ is the modified bessel function of third kind with order 1.

\begin{remark}
	We name the above distribution as bessel distribution by making an analogy to the beta distribution. The density function of the beta case depends on the beta function. In line with this perception, note that expression (\ref{densitybessel}) depends on the bessel function. Hereafter, we denote a random variable $Z$ that follows a bessel distribution by writing $Z \sim \mbox{Bessel}(\alpha,\beta)$.
\end{remark}

The following lemma will be used to obtain the moments of the bessel distribution; see details in \cite{creetal81}.

\begin{lemma}\label{lemma_moments}
	Let $U_1$ and $U_2$ be two random variables with joint moment generation function denoted by $M_{U_1,U_2}(t_1,t_2)=E(\exp\{t_1U_1+t_2U_2\})$. Then, for $j,k\in\mathbb N^*\equiv\{1,2,\cdots\}$, we have that
	\begin{eqnarray*}
		E\left(\dfrac{U_2^j}{U_1^k}\right) \; = \; \dfrac{1}{\Gamma(k)}\int_0^\infty t_1^{k-1}\lim_{t_2\rightarrow0^{-}}\dfrac{\partial^j}{\partial t_2^j}M_{U_1,U_2}(-t_1,t_2) \ dt_1,
	\end{eqnarray*} 
	where $\Gamma(\cdot)$ is the gamma function.
\end{lemma}

\begin{proposition}\label{high-order} If $Z\sim\mbox{Bessel}(\alpha,\beta)$, then the $k$-th moment of $Z$ is given by
	\begin{eqnarray*}
		E(Z^k)=\dfrac{e^{\alpha+\beta}}{\Gamma(k)}\int_0^\infty t_1^{k-1}\lim_{t_2\rightarrow0^{-}}\dfrac{\partial^k}{\partial t_2^k}\exp\left\{-\left(\beta\sqrt{1+2t_1}+\alpha\sqrt{1+2(t_1-t_2)}\right)\right\} \ dt_1, \quad k\in\mathbb N^*.
	\end{eqnarray*}
\end{proposition}
\begin{proof}
	It follows by using Lemma \ref{lemma_moments}. %$\square$
\end{proof}

After some algebra, we obtain that the mean and variance of the bessel distribution are:
\begin{eqnarray*}
	E(Z) \ \equiv \ \mu \ = \ \dfrac{\alpha}{\alpha+\beta} \quad \mbox{and} \quad \mbox{Var}(Z) \ = \ \mu(1-\mu)\dfrac{1-\phi+\phi^2e^\phi Ei(\phi)}{2},
\end{eqnarray*}
where $\phi = \alpha + \beta$ and $Ei(\phi) = \int_1^\infty u^{-1}e^{-\phi u} \ du$ is the exponential integral function, which is implemented in several softwares.

\begin{remark} The two first cumulants are also given in \cite{lijetal05}, where the variance has another equivalent representation based on the incomplete gamma function. On the other hand, no other moments are given in that paper. We highlight the fact that higher-order moments of the bessel distribution can be obtained from Proposition \ref{high-order}. 
\end{remark}

We may consider a reparameterization of the bessel distribution in terms of $\mu\in(0,1)$ and $\phi>0$ defined above, which are the mean and precision parameters. This mean-precision parameterization is of great interest to build a regression model, that is one of the main proposals of the present paper, and it will be considered in Section \ref{sec:reg}. Under this parameterization, we use the notation $Z\sim\mbox{Bessel}(\mu,\phi)$. The associated density function can be written as
\begin{eqnarray}\label{densitybessel_rep}
f(z) \ = \ \dfrac{\mu(1-\mu)\phi e^{\phi}}{\pi [z(1-z)]^{3/2}}\dfrac{K_1\left(\phi\zeta_\mu(z)\right)}{\zeta_\mu(z)}, \quad z\in(0,1),
\end{eqnarray}
where $\zeta_\mu(z) \ = \ \sqrt{1+\dfrac{(z-\mu)^2}{z(1-z)}}$, for $z\in(0,1)$.

A remarkable feature of the beta distribution is that its density function converges to the probability function of a Bernoulli distribution, with success parameter $\mu$, as $\phi\rightarrow0^+$. In contrast with this behavior of the beta case, we conclude the present section by showing that a bessel random variable converges to a continuous random variable as $\phi\rightarrow0^+$.

\begin{proposition}\label{limphi0}
	For a fixed $\mu\in(0,1)$, we have that the reparameterized bessel density function given in (\ref{densitybessel_rep}) satisfies
	$$\lim_{\phi\rightarrow0^+}f(z)=\dfrac{\mu(1-\mu)}{\pi\sqrt{z(1-z)} \ [z(1-z)-(z-\mu)^2]},\quad z\in(0,1).$$
\end{proposition}
\begin{proof}
	By using that $K_1(z)\sim z^{-1}$ for $z\sim0$, we obtain immediately the result. %$\square$
\end{proof}

\section{Regression analysis and EM algorithm}\label{sec:reg}

Here we introduce the bessel regression model and discuss estimation of the parameters. The bessel regression is defined by assuming that ${\bf Z} = (Z_1,\cdots,Z_n)^\top$ are independent random variables with $Z_i\sim\mbox{Bessel}(\mu_i,\phi_i)$, for $i=1,\cdots,n$, where
\begin{eqnarray*}
	\mbox{logit}\, \mu_i \ = \ {\bf x}_i^\top{\boldsymbol\kappa} \quad \mbox{and} \quad \log\phi_i \ = \ {\bf v}_i^\top{\boldsymbol\lambda}.
\end{eqnarray*}
The terms $\boldsymbol\kappa=(\kappa_1,\cdots,\kappa_p)^\top\in\mathbb R^p$ and $\boldsymbol\lambda=(\lambda_1,\cdots,\lambda_q)^\top\in\mathbb R^q$ are vectors of unknown coefficients, which are assumed to be functionally independent. In addition, ${\bf x}_i=(x_{i1}, x_{i2}, \cdots,x_{ip})^\top$ and ${\bf v}_i=(v_{i1},v_{i2},\cdots,v_{iq})^\top$ are observations on $p$ and $q$ known covariates, for $i = 1,\cdots,n$; consider $p + q < n$. The first components $x_{i1}$ and $v_{i1}$ may be equal to $1$ ($\forall \; i$), when intercepts are included in the model. Use $\bf X$ to represent the $n\times p$ matrix with $(i,j)$-th element being $x_{ij}$ and $\bf V$ to denote the $n\times q$ matrix with $(i,j)$-th element being $v_{ij}$.

The justifications for the logit and log link functions above are their practical interpretations and due to the fact that they are default/convenient choices for linking bounded (in the unit interval) and positive parameters to linear predictors. 

Let $\boldsymbol\theta=(\boldsymbol\kappa^\top,\boldsymbol\lambda^\top)^\top$ be the parameter vector. The log-likelihood function is given by
\begin{eqnarray}
	\ell(\boldsymbol\theta) \; \propto \; \sum_{i=1}^n\left\{\log\mu_i+\log(1-\mu_i)+\log\phi_i+\phi_i-\log\zeta_{\mu_i}(z_i)+\log K_1\left(\phi_i\zeta_{\mu_i}(z_i)\right)\right\},\label{loglik}
\end{eqnarray}
where $z_i$ represents the observed value of $Z_i$, for $i=1,\cdots,n$. Note that this log-likelihood function depends on the bessel function. This is a critical aspect creating some major difficulties to numerically find maximum likelihood estimates. With this in mind, we propose an EM algorithm where the associated M-step consists in maximizing a $Q$-function having a simple form.

We now describe with details the proposed EM algorithm. Consider the augmented data $(Z_1,W_1),\cdots,(Z_n,W_n)$, where $Z_1,\cdots,Z_n$ are observable responses and $W_i=Y_{1i}+Y_{2i}$ are latent random variables such that $Z_i = Y_{1i}/W_i$, as indicated by the stochastic representation in (\ref{storep}). Here, all random variables $\{Y_{1i}\}_{i=1}^n$ and $\{Y_{2i}\}_{i=1}^n$ are independent among them, with $Y_{1i} \sim \mbox{IG}(\mu_i\phi_i)$ and $Y_{2i} \sim \mbox{IG}((1-\mu_i)\phi_i)$; therefore, $W_i\sim\mbox{IG}(\phi_i)$, for $i = 1, \cdots, n$.

The complete log-likelihood function is given by
\begin{eqnarray*}
	\ell_c(\bs\theta) \ \propto \ \sum_{i=1}^n\left\{\log\mu_i+\log(1-\mu_i)+2\log\phi_i-\dfrac{\phi_i^2}{2w_i}\left[\dfrac{\mu_i^2}{z_i}+\dfrac{(1-\mu_i)^2}{1-z_i}\right]\right\}.
\end{eqnarray*}
In order to obtain the E-step of the EM algorithm, we need to find the conditional distribution of $W_i$ given $Z_i = z_i$. We will show, in the next proposition, that this conditional model is a generalized inverse-Gaussian (GIG) distribution. A random variable $U$ following a GIG distribution with parameters $a>0$, $b>0$ and $s\in\mathbb R$ (short notation: $U\sim\mbox{GIG}(a,b,s)$) has density function given by $q(u)=\dfrac{(a/b)^{s/2}}{2K_{s}(\sqrt{ab})}u^{s-1}\exp\left\{-\dfrac{1}{2}\left(au+\dfrac{b}{u}\right)\right\}$, for $u>0$. See \cite{kou14} for more details. 

\begin{proposition} The conditional density function of $W_i$ given $Z_i=z_i$ is
	\begin{eqnarray*}\label{cond_df}
		f(w_i|z_i)=\dfrac{\phi_i\zeta_{\mu_i}(z_i)}{2 K_{-1}\left(\zeta_{\mu_i}(z_i)\right)}w_i^{-2}\exp\left\{-\dfrac{1}{2}\left[\dfrac{\phi_i^2}{w_i}\left(\dfrac{\mu_i^2}{z_i}+\dfrac{(1-\mu_i)^2}{1-z_i}\right)+w_i\right]\right\}, \quad w>0.
	\end{eqnarray*}
	In other words, we have that $W_i|Z_i=z_i\sim\mbox{GIG}(1,\phi_i^2\zeta^2_{\mu_i}(z_i),-1)$.
\end{proposition}
\begin{proof}
	It follows immediately by using that $f(w_i|z_i)\propto f(z_i|w_i)f(w_i)$. Hence, we can identify the kernel of a GIG density function with the parameters stated in the proposition. %$\square$
\end{proof}

Using the previous proposition, one can find important conditional moments for the E-step of the algorithm. The next corollary indicates this result.

\begin{corollary}\label{Estep}
	For $s\in\mathbb R$, we have that \ $E\left(W_i^s|Z_i=z_i;\bs\theta\right)=\phi_i^s\zeta^s_{\mu_i}(z_i)
	\dfrac{ K_{s-1}\left(\phi_i\zeta_{\mu_i}(z_i)\right)}{ K_{-1}\left(\phi_i\zeta_{\mu_i}(z_i)\right)}$.
\end{corollary}

Having presented the conditional density function and moments of latent random variables ${\bf W} = (W_1, \cdots, W_n)$ given the observed data ${\bf Z}=(Z_1,\cdots,Z_n)$, we are now ready to determine the $Q$-function of the EM algorithm. This element is defined by $Q(\bs\theta;\theta^{(r)})\equiv E(\ell_c(\bs\theta)|{\bf Z};\bs\theta^{(r)})$, where $\bs\theta^{(r)}$ is the EM-based estimate of $\bs\theta$ in the $r$-th loop of the algorithm, for $r\in\mathbb N$. More specifically, the $Q$-function is as follows
\begin{eqnarray}
	Q(\bs\theta;\theta^{(r)}) \ \propto \ \sum_{i=1}^n \left\{\log\mu_i+\log(1-\mu_i)+2\log\phi_i+\phi_i-\dfrac{1}{2}\psi_i(\theta^{(r)})\phi_i^2\left(\dfrac{\mu_i^2}{z_i}+\dfrac{(1-\mu_i)^2}{1-z_i}\right)\right\},\label{QFunction}
\end{eqnarray}
where $\psi_i(\bs\theta^{(r)})=E\left(W_i^{-1}|Z_i=z_i;\bs\theta^{(r)}\right)$ can be obtained from Corollary \ref{Estep}, for $i=1,\cdots,n$.

Now the discussion is focused on the M-step of the algorithm. The components of the score function $U(\bs\theta;\theta^{(r)}) \ = \ \partial Q(\bs\theta;\theta^{(r)})/ \partial\bs\theta$ are
\begin{eqnarray*}
	\dfrac{\partial Q(\bs\theta;\theta^{(r)})}{\partial\bs\kappa_j}=\sum_{i=1}^n\left\{1-2\mu_i+\dfrac{1}{2}\psi_i(\bs\theta^{(r)})\phi_i^2\mu_i(1-\mu_i)\dfrac{z_i-\mu_i}{z_i(1-z_i)}\right\}x_{ij},\quad j=1,\cdots,p,
\end{eqnarray*}
and
\begin{eqnarray*}
	\dfrac{\partial Q(\bs\theta;\theta^{(r)})}{\partial\bs\lambda_l}=\sum_{i=1}^n\left\{2+\phi_i
	-\psi_i(\bs\theta^{(r)})\phi_i^2\left(1+\dfrac{(z_i-\mu_i)^2}{z_i(1-z_i)}\right)\right\}v_{il},\quad l=1,\cdots,q.
\end{eqnarray*}

In order to find the EM-based estimates of the parameters in the $r$-th loop of the algorithm, the $Q$-function must be maximized with respect to $\bs\theta$. Since there is no explicit form for the estimators in each loop, we may use some optimization procedure for this task, for example, Newton-Raphson \citep{atk89} and BFGS \citep{fle00}. The score function can be used in these procedures, otherwise, numerical gradients are required.

We now provide a more detailed description of the EM algorithm. Let $\bs\theta^{(0)}$ be the initial guess of $\bs \theta$. Guidelines for choosing initial guesses of the parameters in the bessel regression model are presented ahead in Section \ref{sec:sim}. In the E-step of the algorithm, the conditional expectations $\psi_1(\bs\theta^{(r)}),\cdots,\psi_n(\bs\theta^{(r)})$ are updated with the previous ($r$-th loop) EM-estimate of $\bs\theta$. Next, the M-step is applied to maximize the $Q$-function and obtain the $(r+1)$-th estimate of the parameters, which is denoted by $\bs\theta^{(r+1)}$. In this stage, the algorithm verifies the convergence criterion $||\bs\theta^{(r+1)}-\bs\theta^{(r)}|| \ / \ ||\bs\theta^{(r)}|| < \epsilon$, for some pre-specified $\epsilon > 0$. If this criterion is met, the EM-based estimate of $\bs\theta$ is set to be $\bs\theta^{(r+1)}$. Otherwise, $\bs\theta^{(r)}$ is replaced by $\bs\theta^{(r+1)}$ and the mentioned steps of the algorithm are repeated.

The observed information matrix can be computed from an EM-approach. In fact, the observed information matrix, denoted by ${\bf I}_n(\bs\theta)$, was obtained by \cite{lou82} and is given by
\begin{eqnarray}\label{infmatrix}
{\bf I}_n(\bs\theta)=E\left(-\dfrac{\partial^2\ell_c(\bs\theta)}{\partial\bs\theta\partial\bs\theta^\top}\Big|{\bf Z}\right)-E\left(\dfrac{\partial\ell_c(\bs\theta)}{\partial\bs\theta}\dfrac{\partial\ell_c(\bs\theta)}{\partial\bs\theta}^\top\Big|{\bf Z}\right).
\end{eqnarray}
The elements of the information matrix (\ref{infmatrix}) are presented in the Appendix. The standard errors of the parameters can be obtained through the estimated matrix ${\bf I}_n(\widehat{\bs\theta})$, where $\widehat{\bs\theta}$ is the EM-estimator of $\bs\theta$. Note that the $Q$ function given in equation \eqref{QFunction} is continuous on both $\bs\theta$ and $\bs\theta^{(r)}$. Thus, by using Theorem 2 in \cite{wu83}, we conclude that any limit point of $\bs\theta^{(r)}$ is a stationary point of the likelihood function \eqref{loglik}. We now assume that the usual regularity conditions \cite[Conditions (a)-(d) in p. 281]{coxhinkley74} hold for the log-likehood function \eqref{loglik} and also that, for large $n$, the likelihood function \eqref{loglik} admits a unique maximum. Thus, under these conditions, the log-likelihood function has only one stationary point, thus ensuring that $\bs\theta^{(r)}$ has only one limit point, which is given by the unique maximum likelihood estimator (MLE) of $\bs\theta$, namely $\widehat{\bs\theta}$. Moreover, from the asymptotic normality of the MLEs and the observed information matrix given in \eqref{infmatrix}, we have that $\sqrt{n}(\widehat{\bs\theta}-\bs\theta)\stackrel{d}{\longrightarrow}N(0,\bs\Sigma^{-1})$, where $\bs\Sigma$ is the limit in probability of $n^{-1}{\bf I}_n(\bs\theta)$, as $n\rightarrow\infty$. With this, asymptotic confidence intervals may be constructed for the model parameters.

All functions and programs to generate data and to fit the models in the present paper were implemented through the \texttt{R} \citep{sofR} programming language.

\section{Discrimination test between bessel and beta regressions}\label{sec:discrimination}

Let $Z_1,\cdots, Z_n$ be independent random variables having $E(Z_i) \ = \ \mu_i$ and $Var(Z_i) \ = \ \mu_i(1-\mu_i) \ g(\phi_i)$, where $g(\cdot)$ is a continuous, monotone and unknown function. Consider the hypotheses: 
$$
 \mathcal H_{beta}: \quad Z_i \sim \hbox{Beta}(\mu_i,\phi_i) \;\;  \forall i,\qquad \hbox{and}\qquad \mathcal H_{bessel}:\quad Z_i \sim \hbox{Bessel}(\mu_i,\phi_i) \;\; \forall i.
$$
If $Z_i \sim \hbox{Beta}(\mu_i,\phi_i)$, then 
\begin{equation}
g(\phi_i) \equiv g_{beta}(\phi_i) =  1/(1+\phi_i). \label{gphi_bet}    
\end{equation}
If $Z_i \sim \hbox{Bessel}(\mu_i,\phi_i)$, then 
\begin{equation}
g(\phi_i) \equiv g_{bessel}(\phi_i) = \frac{1-\phi_i+\phi_i^2\exp\{\phi_i\}Ei(\phi_i)}{2}. \label{gphi_bes}
\end{equation} 
In each of these cases, $g(\cdot)$ is strictly decreasing. Furthermore, for every $\phi\geq 0$, we write 
$$
1-\phi+\phi^2\exp\{\phi\} Ei(\phi) \leq 1 \; \Rightarrow \;  g_{bessel}(\phi) \leq \frac{1}{2}.
$$
Thus, if $Z_i \sim \hbox{Bessel}(\mu_i,\phi_i)$, we obtain that $Var(Z_i) \leq \mu_i(1-\mu_i)/2$. On the other hand, if $Z_i \sim \hbox{Beta}(\mu_i,\phi_i)$, one can write $Var(Z_i) \leq \mu_i(1-\mu_i)$, since $g_{beta}(\phi)\leq 1$ and $g_{beta}(0)=1$.

The previous result indicates that the bessel distribution is more suitable to underdispersed bounded data. From our experience in dealing with the beta regression, the underdispersed scenario is highly common in bounded data sets (typically rates or proportions). The justification lies in the fact that the data are bounded. In the beta case, the maximum dispersion is achieved as a limiting case, namely, a discrete distribution concentrating mass 1/2 on 0 and 1. This aspect configures a constrast with respect to the bessel case (see Proposition \ref{limphi0}). The remaining ``high'' variance scenarios are also in this fashion, that is, concentration of masses around 0 and 1. Note that these cases are uncommon from a practitioner's point of view (even though they occur). Furthermore, as described in \cite{bss17}, the marginal log-likelihood function with respect to $\phi$ (without covariates) becomes flat fairly quickly, thus providing poor estimates for large values of the precision parameter, i.e., for small variances. So, the beta distribution is not suitable for underdispersed data sets. It is also remarkable that, besides 0 or 1, the beta distribution cannot ``concentrate'' around any other point. So it tends to spread the data along the interval, providing possible concentration on one (or both) of the endpoints. As a result, for underdispersed data with values away from the endpoints, the beta regression will probably provide a poor fit. We expect that the bessel regression, introduced in this paper, can show better performance in this situation.

It is important to note that this does not mean that the variance obtained through the bessel regression will always be lower than the variance obtained through a beta regression on the same dataset. The important message here is that the structure of the bessel regression is more suitable to fit underdispersed data sets than the beta regression, as discussed above.

Our goal now is to use the difference between variance structures of the beta and bessel distributions to determine which one should be considered for a given data set having an unknown bounded distribution. To this end, since the distribution of $Z_i$ is unknown, we will take advantage of the consistency of the quasi-likelihood estimators for a very large class of distributions to provide a test to determine which (if any) should be used.

As before, let $\mu_i = \mu_i({\boldsymbol\kappa}) =  \dfrac{\exp\{{\bf x}_i^\top{\boldsymbol\kappa}\}}{{1+\exp\{{\bf x}_i^\top{\boldsymbol\kappa}\}}}$ and  $\phi_i = \exp\{{\bf v}_i^\top{\boldsymbol\lambda}\}$, for $i=1,\cdots,n$. Now, assume $\widetilde{\boldsymbol\kappa}$ and $\widetilde{\boldsymbol\mu}$ to be the quasi-likelihood estimators of ${\boldsymbol\kappa}$ and ${\boldsymbol\mu}$, respectively. In addition, define
$$
 U_{QL}({\boldsymbol\kappa}) = \sum_{i=1}^n \dfrac{(y_i-\mu_i({\boldsymbol\kappa})){\bar\mu}_i(\bs\kappa)}{\sqrt{\mu_i({\boldsymbol\kappa})(1-\mu_i({\boldsymbol\kappa}))}}\cdot{\bf x}_i=\sum_{i=1}^n (y_i-\mu_i({\boldsymbol\kappa}))\sqrt{\mu_i({\boldsymbol\kappa})(1-\mu_i({\boldsymbol\kappa}))}\cdot{\bf x}_i,
$$
where $\bar\mu_i(\bs\kappa)=\mu_i({\boldsymbol\kappa})(1-\mu_i({\boldsymbol\kappa}))$, for $i=1,\ldots,n$. The term $\widetilde{\boldsymbol\kappa}$ is the solution of the system of equations $U_{QL}(\widetilde{\boldsymbol\kappa}) = 0.$ Note that $\widetilde{\boldsymbol\kappa}$ does not depend on the estimated precision parameter. Furthermore, under usual regularity conditions,
$\widetilde{\boldsymbol\kappa} \stackrel{p}{\longrightarrow} {\boldsymbol\kappa}$ as $n\to\infty$. Since ${\bf x}_i$ is assumed deterministic and fixed for any sample size $n$, we write for each $i$
\begin{equation}\label{muconv}
\widetilde{\mu}_i \stackrel{p}{\longrightarrow} \mu_i,
\end{equation}
as $n\to\infty$. As can be seen, the variables $Z_1,\cdots,Z_n$ are independent and bounded, therefore, we may apply Kolmogorov's Strong Law of Large Numbers for independent and non-identically distributed random variables to conclude that
$$
 \sum_{i=1}^n \frac{Z_i^2}{n} - \sum_{i=1}^n \frac{\mu_i(1-\mu_i)g(\phi_i)+\mu_i^2}{n} \; \stackrel{a.s.} {\longrightarrow} \; 0.
$$
Hence, it follows from \eqref{muconv} that
\begin{eqnarray}\label{conv1}
\sum_{i=1}^n \frac{Z_i^2}{n} - \sum_{i=1}^n \frac{\widetilde{\mu}_i(1-\widetilde{\mu}_i)g(\phi_i)+\widetilde{\mu}_i^2}{n} \; \stackrel{p}{\longrightarrow} \; 0.
\end{eqnarray}

Now, an EM-scheme is considered for estimating $\boldsymbol\lambda$ (and thus estimating $\phi_1,\cdots,\phi_n$) under ${\cal H}_{bessel}$ and ${\cal H}_{beta}$. This is done by keeping $\widetilde{\mu}_i$ fixed (where we will use the consistent estimate given by the quasi-likelihood estimator) instead of $\mu_i$, for $i=1,\cdots,n$. Let $\widetilde{\phi}_i^{bessel}$ and $\widetilde{\phi}_i^{beta}$ be the EM-estimates under the bessel and beta regressions. In these cases, the EM algorithm is given as before but assuming $\widetilde{\mu}_1,\cdots,\widetilde{\mu}_n$ fixed. Under the hypothesis ${\cal H}_{bessel}$, we have
\begin{equation}\label{convphi}
\widetilde{\phi}_i^{bessel} \stackrel{p}{\longrightarrow} \phi_i \quad \forall i.
\end{equation}
On the other hand, for ${\cal H}_{beta}$ we write
$$
 \widetilde{\phi}_i^{beta} \stackrel{p}{\longrightarrow} \phi_i \quad \forall i.
$$
Since the function $g_{bessel}(\cdot)$ is continuous, we may apply \eqref{convphi} in \eqref{conv1} to conclude, under ${\cal H}_{bessel}$, that   
\begin{equation*}\label{convbessel1}
\sum_{i=1}^n \frac{Z_i^2}{n} - \sum_{i=1}^n \frac{\widetilde{\mu}_i(1-\widetilde{\mu}_i)g_{bessel}(\widetilde{\phi}_i^{bessel})+\widetilde{\mu}_i^2}{n} \stackrel{p}{\longrightarrow} 0.
\end{equation*}
Analogously, under ${\cal H}_{beta}$, we write
\begin{equation*}\label{convbessel2}
\sum_{i=1}^n \frac{Z_i^2}{n} - \sum_{i=1}^n \frac{\widetilde{\mu}_i(1-\widetilde{\mu}_i)g_{beta}(\widetilde{\phi}_i^{beta})+\widetilde{\mu}_i^2}{n} \; \stackrel{p}{\longrightarrow} \; 0.
\end{equation*}
Finally, our criterion for discrimination is defined as follows. The beta regression should be used when $\sum_{i=1}^n Z_i^2/n \geq \sum_{i=1}^n \left(\widetilde{\mu}_i(1-\widetilde{\mu}_i)/2 +\widetilde{\mu}_i^2\right)$. Otherwise, compute:
$$
D_{bessel} \; = \; \sum_{i=1}^n \frac{Z_i^2}{n} - \sum_{i=1}^n \frac{\widetilde{\mu}_i(1-\widetilde{\mu}_i)g_{bessel}(\widetilde{\phi}_i^{bessel})+\widetilde{\mu}_i^2}{n}
$$
and 
$$
D_{beta} \; = \; \sum_{i=1}^n \frac{Z_i^2}{n} - \sum_{i=1}^n \frac{\widetilde{\mu}_i(1-\widetilde{\mu}_i)g_{beta}(\widetilde{\phi}_i^{beta})+\widetilde{\mu}_i^2}{n}.
$$
If $|D_{bessel}| \leq |D_{beta}|$, select the bessel regression introduced in this paper for the data analysis. Otherwise, apply the beta regression.

Note that we choose the bessel distribution in a tie (which will rarely occur) due to the flatness of the log-likelihood function based on the beta distribution with respect to the precision parameter in an underdispersed model. Hereafter in this paper, we denote the test described in this section by DBB criterion, which stands for ``Discrimination between Bessel and Beta models''. In the end of the next section, we develop a short simulation study to investigate its classification performance for different sample sizes. 

\section{Simulation results}\label{sec:sim}

This section shows the results related to a simulation study exploring the performance of the proposed bessel regression. The main goal is to evaluate how well the model can handle data generated from the bessel regression setting itself. Good estimates indicate that the algorithm is correctly implemented and thus validate results discussed ahead in a real application. Analyses involving data misspecification are developed in the next section. 

Synthetic data sets are generated from a bessel regression model assuming an intercept and two covariates in ${\bf x}_i$ and ${\bf v}_i$. The first covariate is binary and generated from the $\mbox{Bernoulli}(0.5)$. The second one is obtained from the $\mbox{U}(-1.0,1.0)$. The values of the covariates linked to $\mu_i$ and $\phi_i$ are not the same. The true coefficients have the following configuration: ${\boldsymbol\kappa} = (0.5, -0.5, 1.0)^\top$ and ${\boldsymbol\lambda} = (1.5, 1.0, -0.5)^\top$. Different sample sizes are explored in this analysis, they are: $n = 50$, $100$, $200$ and $500$. A Monte Carlo (MC) structure is considered here with $1000$ data sets replicated for each sample size $n$. The steps to generate data are summarized as follows: ($i$) choose $n$ and generate the matrices ${\bf X}$ and ${\bf V}$, ($ii$) compute $\mu_i = \exp\{{\bf x}_i^\top {\boldsymbol\kappa}\} / (1+\exp\{{\bf x}_i^\top {\boldsymbol\kappa}\})$ and $\phi_i = \exp\{{\bf v}_i^\top {\boldsymbol\lambda}\}$, ($iii$) set $r = 1$ to indicate the first data set in the MC scheme, ($iv$) generate $Y_{1i}$ from the inverse-Gaussian distribution with mean = variance = $\mu_i \phi_i$, ($v$) generate $Y_{2i}$ from the inverse-Gaussian distribution with mean = variance = $\phi_i (1-\mu_i)$, ($vi$) calculate the response $Z_i = Y_{1i} / (Y_{1i} + Y_{2i})$, which follows the $\mbox{Bessel}(\mu_i, \phi_i)$ distribution, ($vii$) return to step 3 and update the iteration number to be $r+1$. In order to generate from the inverse-Gaussian distribution, we use the \texttt{R} package \texttt{statmod} \citep{gin16}. 

In the EM algorithm, the initial values for ${\boldsymbol \kappa}$ and ${\boldsymbol \lambda}$ are the default choices implemented in the \texttt{R} package \texttt{betareg} \citep{cri10} to fit a beta regression. As described in the package documentation, starting values are obtained from an auxiliary linear regression applied to the transformed response. In this case, non-zero values are determined for ${\boldsymbol \kappa}$ and the intercept $\lambda_1$. The choice $0$ is set for all remaining components in ${\boldsymbol \lambda}$. Regarding the intercept $\lambda_1$, since the bessel and beta models differ in terms of their precision parameter, the following adaptation is required for the bessel case: $\widetilde\lambda_1^{(0)} = \ln(2) + \ln(1+\exp\{\lambda_1^{(0)}\})$, where $\lambda_1^{(0)}$ is the starting point from \texttt{betareg}.

The maximization step of the EM algorithm is performed through the \texttt{R} general purpose optimization command \texttt{optim}. In this case, the \texttt{BFGS} method \citep{fle00} is applied to all scenarios investigated in our study. The tolerance value, determining the convergence criterion of the EM algorithm, is set to be $\epsilon = 10^{-5}$. Since we run the algorithm independently for each sample in the MC scheme, the \texttt{R} package \texttt{snowfall} \citep{kna15} is applied for faster results through parallel computing.

\begin{figure}[!h]
	\centering
	\includegraphics[scale=0.23]{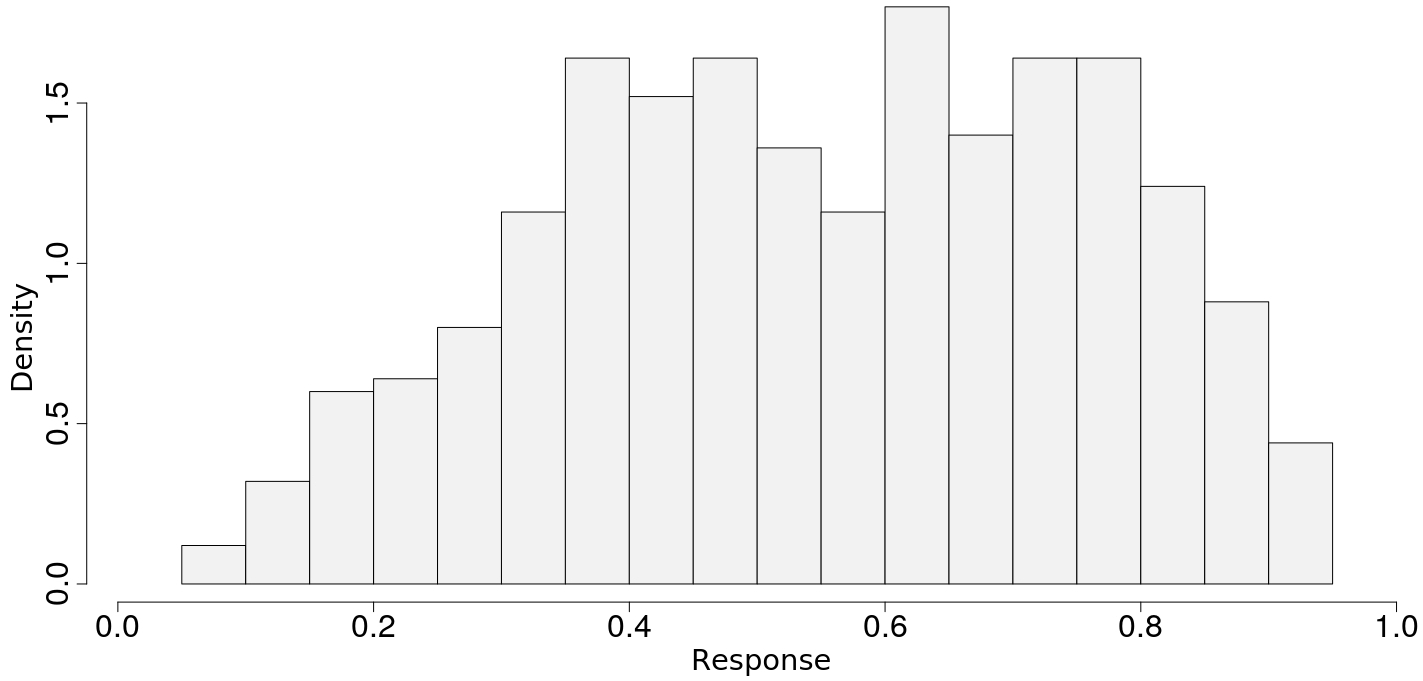} 
	\caption{Histogram displaying the behavior of the response variable generated in the first MC sample ($n$ = 500).} \label{histy}
\end{figure}

%\vspace{-10pt}

Figure \ref{histy} shows a histogram representing the behavior of the response variable generated in the first MC sample with $n = 500$. Note that the chosen configuration of true coefficients and the generated values of regressors for each sample unit provide a bessel distributed response well dispersed in the interval ($0$,$1$). Higher frequency is observed between 0.3--0.8, i.e. the distribution is not concentrated in the lower or upper border. This same shape is seen in all MC samples.

The boxplots in Figure \ref{fig_est} represent the distribution of the MC estimates for each coefficient and each sample size. Note that these graphs indicate a symmetric behavior with median and mean (small circle) being quite similar. In terms of inference, the coefficients in ${\boldsymbol \kappa}$ are well estimated; the boxplots are centered around the true value (gray horizontal lines). This aspect is also observed for $\lambda_2$ and $\lambda_3$. The intercept $\lambda_1$, linked to $\phi_i$, is slightly overestimated for small sample sizes ($n = 50$ and $100$). This bias is not observed for large samples ($n = 200$ and $500$). As expected, the variability exhibited by the boxplots reduces as $n$ increases. In addition, the variabilities related to ${\boldsymbol \kappa}$ are smaller than those for ${\boldsymbol \lambda}$. 

Figure \ref{fig_se} shows boxplots summarizing the standard errors obtained via information matrix (\ref{infmatrix}) when fitting the bessel regression to each MC sample. The small solid circles indicate the mean of the standard errors represented in the graphs. The large gray circles are the standard deviations of the MC estimates forming the boxplots in Figure \ref{fig_est}. In a scenario where the code is well implemented, one should expect similar values of the average standard errors (small solid circles) and the MC standard deviations (large gray circles). This is observed in almost all cases exhibited in Figure \ref{fig_se}. A small gap between these measurements can be detected for the smallest sample size ($n = 50$). This is more evident for the coefficients in ${\boldsymbol \lambda}$. The smaller variability related to ${\boldsymbol \kappa}$ (lower boxplots for a fixed $n$) is also obvious in this graph. The effect of the sample size can also be emphasized here (variance decreases as $n$ increases).

The results discussed in this section suggest that the proposed bessel regression fitted via EM algorithm is behaving well for different sample sizes $n$. In general, it is easier to estimate the coefficients linked to the response mean $\mu_i$; we see lower standard error and bias for ${\boldsymbol \kappa}$. When $n = 50$, a very small deviation from the true value can be noted for the intercept $\lambda_1$. This issue vanishes as $n$ increases. In this section, we investigate the performance of the bessel regression in a favorable scenario, where the data is originated from the bessel distribution. The next section evaluates the robustness of this model under the situation of data misspecification. 

We now develop a short simulation study to evaluate the performance of the DBB criterion applied to synthetic data sets. The group of $1000$ data sets to be examined here (for each sample size $n = 50$, $100$, $200$ and $500$) is the same one obtained from the bessel regression setting as previously described in the analysis of Figures \ref{fig_est} and \ref{fig_se}. An extra step is necessary to generate other $1000$ data sets from the beta regression structure. Consider, in step ($vi$) of the previous scheme, the simulation $Z_i \sim \mbox{Beta}(\mu_i,\phi_i)$. The same covariates are used for all MC samples in both models.  

\begin{figure}[!h]
	\centering
	\includegraphics[scale=0.30]{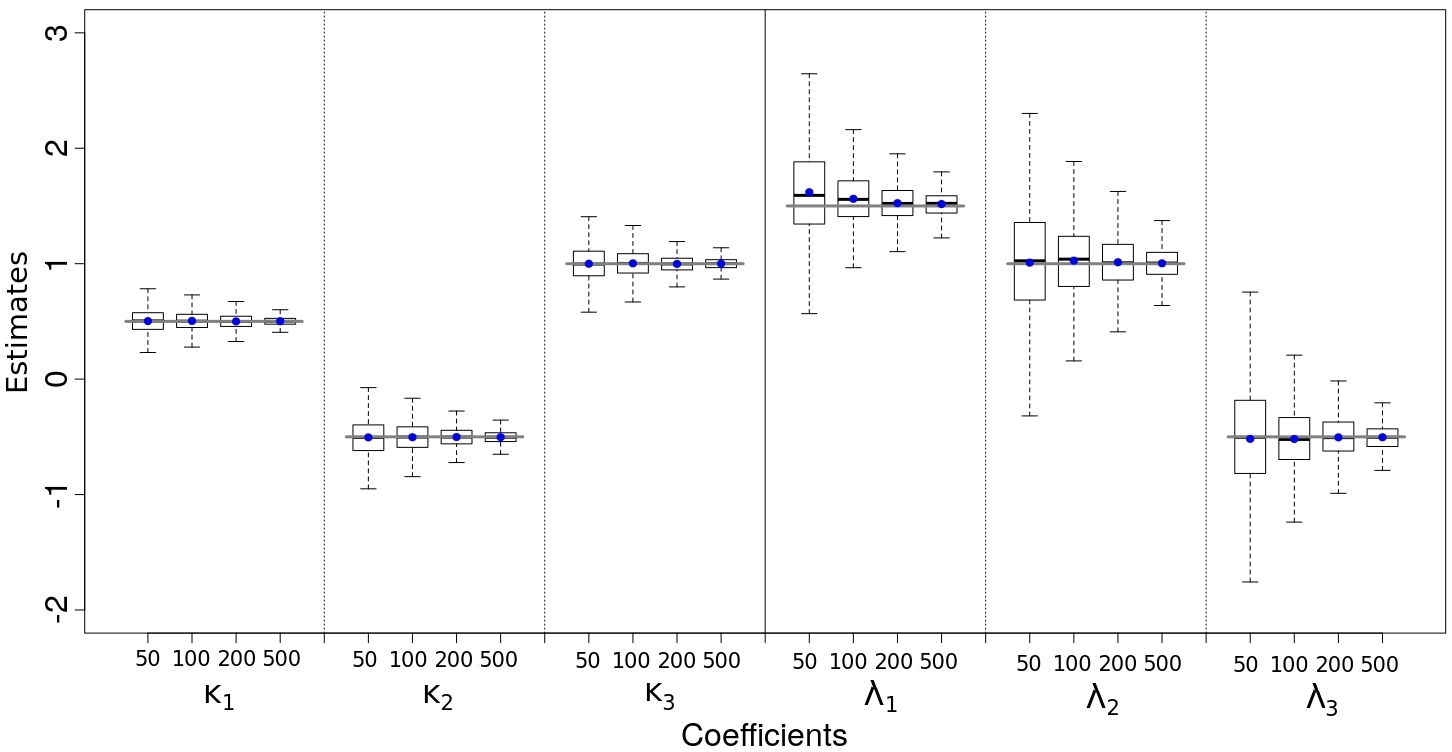}  
	\caption{Boxplots summarizing the MC results for each sample size $n$ and each coefficient. The horizontal gray line indicates the true value of the parameter. The small solid circles represent the MC mean.} \label{fig_est}
\end{figure}

\begin{figure}[!h]
	\centering
	\includegraphics[scale=0.30]{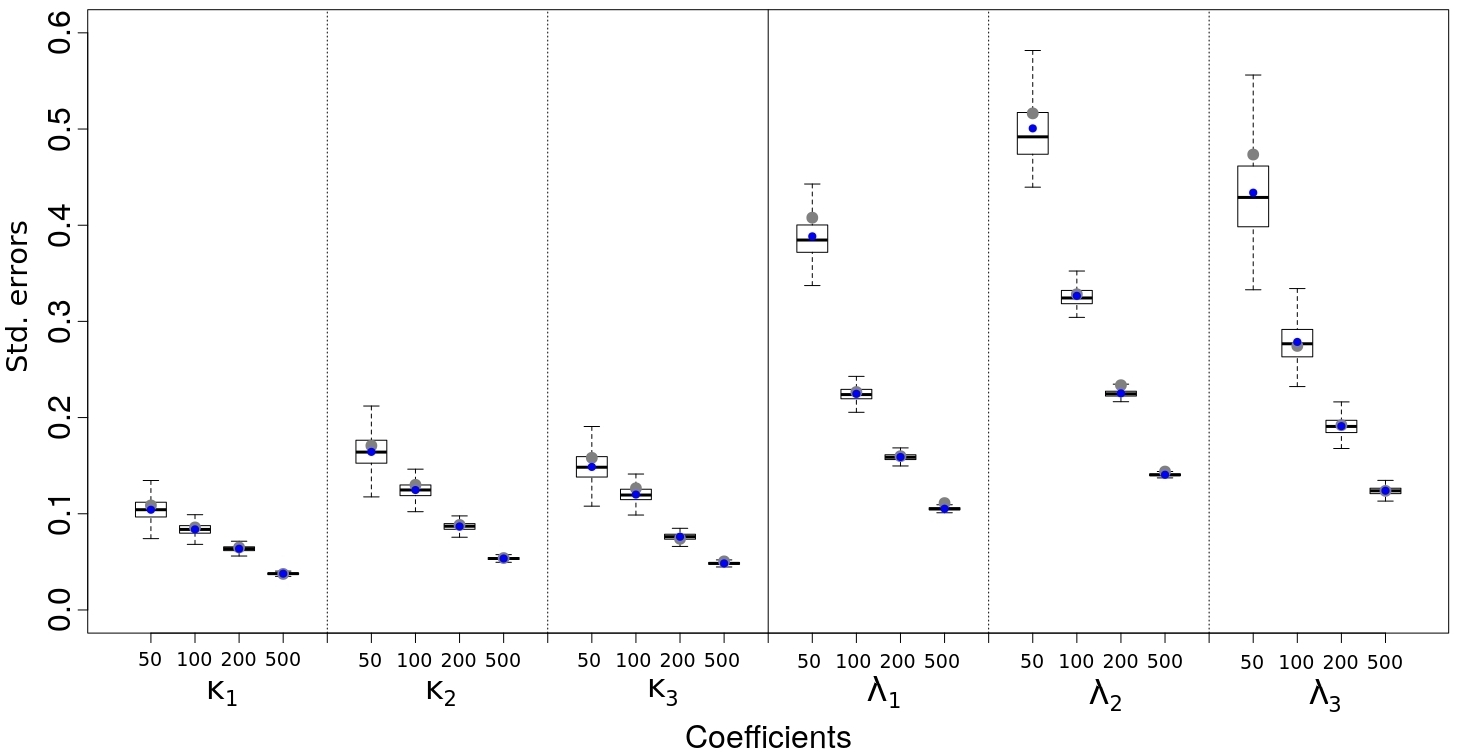}  
	\caption{Boxplots summarizing the standard errors obtained via information matrix (\ref{infmatrix}) for each MC replication. The small solid circles represent the mean of the standard errors forming the boxplots. The large gray circles indicate the MC standard deviations of the estimates for each parameter.} \label{fig_se}
\end{figure}

\begin{table}[!h]
	\centering
	\tabcolsep=2.5pt
	\begin{tabular}{l|rrrr}
		\hline
		Data generator model & \multicolumn{1}{c}{$n = 50$} & \multicolumn{1}{c}{$n = 100$} & \multicolumn{1}{c}{$n = 200$} & \multicolumn{1}{c}{$n = 500$} \\
		\hline
		bessel regression & 67.8 & 71.1 & 79.1 & 88.0 \\
		beta regression   & 37.3 & 23.7 & 10.3 &  2.5 \\
		\hline 
	\end{tabular}
	\caption{Percentages of data sets receiving the indication of bessel regression according to the DBB criterion proposed in Section \ref{sec:discrimination}. Values are calculated with respect to the universe of $1000$ MC replications for each sample size $n$ and each generator model.} \label{tab:test}
\end{table}

Table \ref{tab:test} shows results from a simulation study where MC replications are submitted to the proposed discrimination test. The reported percentages represent how often the bessel regression is chosen as the most appropriate model for the synthetic data sets. As can be seen, high percentages are observed when the data is indeed originated from the bessel model and low values are obtained when the beta regression is the data generator. Another important aspect shown in Table \ref{tab:test} is the fact that the number of correct classifications tend to increase as the sample size $n$ increases.   

\section{Robustness under misspecification}\label{sec:miss}

This simulation involves data sets generated from a beta regression setting. In fact, the response variable is contaminated with a small percentage of values originated from a beta distribution with mean $0.2$. The main goal is to compare the performances of the bessel and beta regressions to fit the infected data sets. Note that some advantage is given to the beta regression, since this is the generator model for the majority of the sample observations. In order to generate the data, we assume again one intercept and two covariates in ${\bf x}_i$; first regressor is binary from $\mbox{Bernoulli}(0.5)$ and the second one is continuous from $\mbox{U}(-1.0,1.0)$. Since the bessel and beta regressions are not comparable in terms of ${\boldsymbol \lambda}$, we simplify the modeling by avoiding covariates linked to $\phi_i$; i.e. ${\boldsymbol \lambda} = \lambda_1$ is an intercept, ${\bf V}$ is a vector of 1's and $\phi_i = \exp\{\lambda_1\}$ for all $i$. The true values of the coefficients are: ${\boldsymbol\kappa} = (0.5, -0.5, 1.0)^\top$ and $\lambda_1 = \ln(5)$. The MC scheme is also explored here with $1000$ replications for each sample size $n = 50$, $100$, $200$ and $500$. The following steps are considered to obtain the response variable: ($i$) choose $n$ and generate the matrix of covariates ${\bf X}$; ($ii$) choose the probability of contamination $p_c \in \{0$, $0.01$, $0.02$, $\cdots$, $0.10\}$; ($iii$) set $r = 1$ to indicate the first MC data set; ($iv$) generate the contamination indicator $C_i \sim \mbox{Bernoulli}(p_c)$ for $i = 1, \cdots, n$; ($v$) if $C_i = 0$, compute $\mu_i = \exp\{{\bf x}_i^\top {\boldsymbol\kappa}\} / (1+\exp\{{\bf x}_i^\top {\boldsymbol\kappa}\})$ and keep $\phi_i = 5$, otherwise set $\mu_i = 0.2$ and make $\phi_i = 50$ to reduce the variance of the beta distribution being the source of contamination; ($vi$) generate the response $Z_i$ from the beta distribution with shape parameters $\mu_i \phi_i$ and $(1-\mu_i) \phi_i$; ($vii$) return to step 3 and update the iteration number to be $r+1$. 

\begin{figure}[!h]
	$$
	\begin{array}{cc}
	\mbox{(a)} & \mbox{(b)} \\
	\includegraphics[scale=0.22]{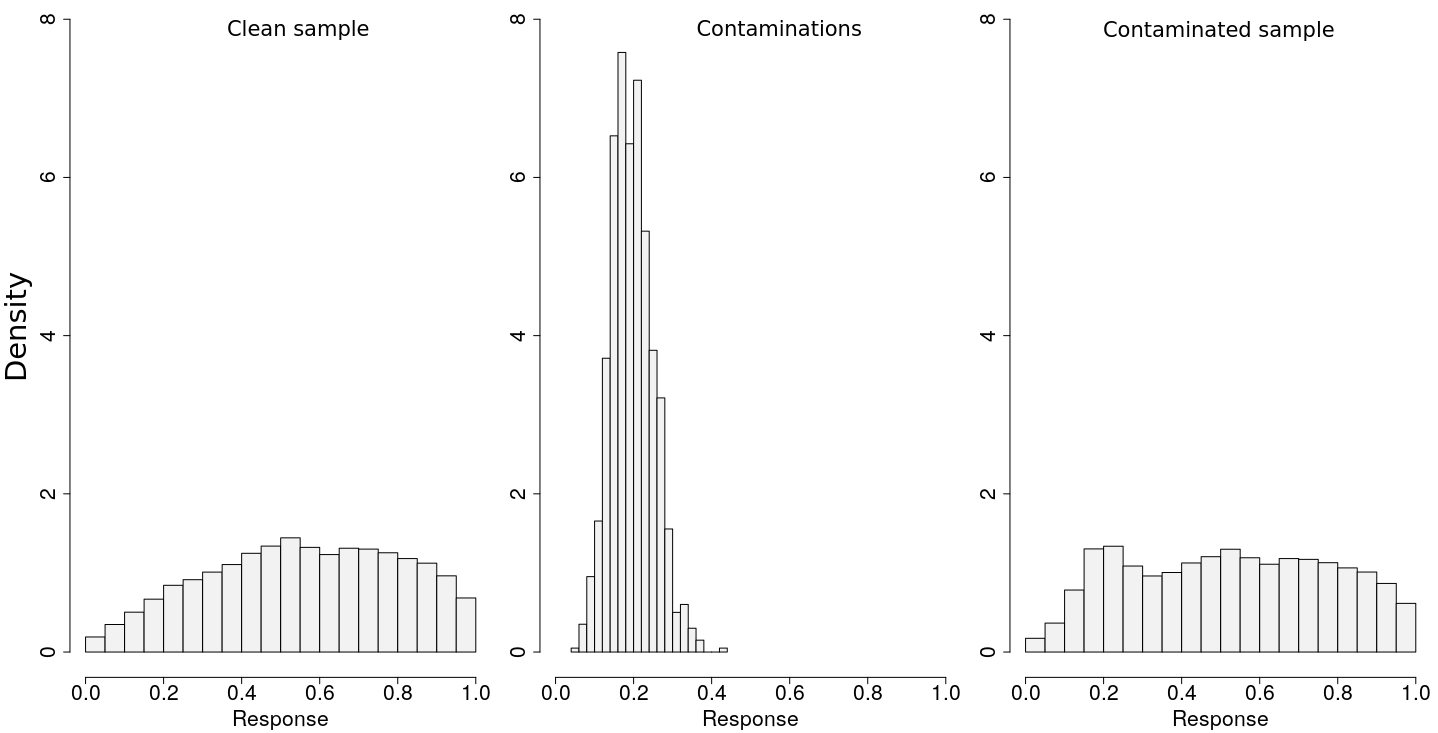} & 
	\hspace{0.5cm} \includegraphics[scale=0.22]{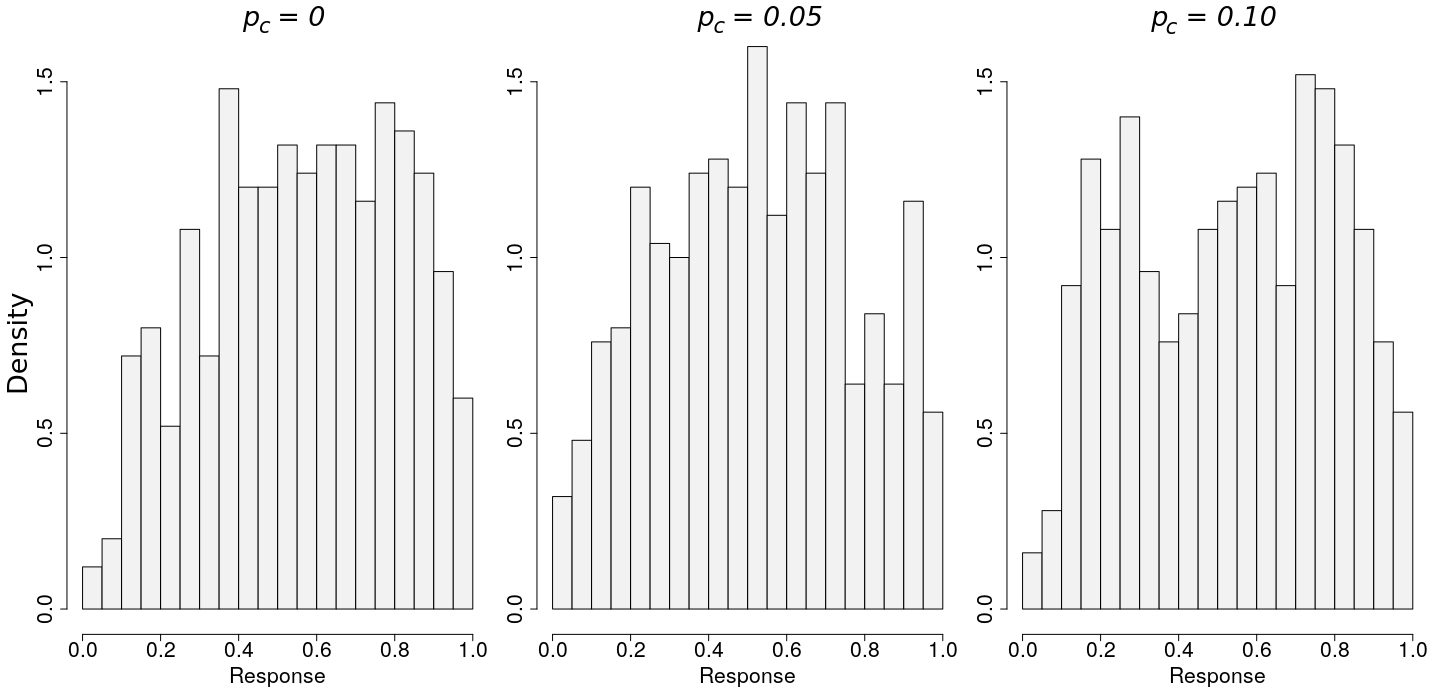} \\
	\end{array}
	$$  
	\caption{Histograms representing the behavior of the synthetic data. Panel (a) shows (assume $n = 10{,}000$): a clean sample generated from the beta regression without contamination, a sample of contaminations generated from a beta model with mean $0.2$ and a sample of response values originated from a mixture of the first two cases ($p_c = 0.10$). Panel (b) displays the first MC sample ($n$ = 500) generated with $p_c = 0$, $0.05$ and $0.10$.} \label{histy_con}
\end{figure}

Figure \ref{histy_con} shows six histograms to illustrate the behavior of the synthetic data. Panel (a) is built with a large sample size ($n = 10{,}000$) to allow a clear visualization of the distribution of response values in three scenarios. The first one is a clean sample without any contamination and originated from the main beta regression model. The second histogram represents a sample of contaminations from the beta model with mean $0.2$. Finally, the third case is related to a sample of responses with approximately $10\%$ ($p_c = 0.10$) of infected values. Note that, when comparing the first and third graphs, a small hill can be detected around $0.2$ in the infected case. This is clearly a deviation from the generator beta model leading to a bimodal configuration, which is expected to be better accommodated through the flexible bessel regression. The empirical mean of the samples presented in Panel (a) are 0.559, 0.197 and 0.523, respectively; due to the presence of contaminations, the mean related to the third graph is slightly smaller than that from the clean sample. The graphs in Panel (b) exhibit the first MC sample ($n = 500$) generated for the choices: $p_c = 0$ (no infection) $p_c = 0.05$ ($\approx 5 \%$ of infection) and $p_c = 0.10$ ($\approx 10\%$ of infection). The third graph in the panel clearly shows a local mode around $0.2$ corresponding to the contaminations. Naturally, the strength of this local mode is reduced for small $p_c$. The reader should also bear in mind the fact that the number of contaminations depends on the sample size $n$, therefore, the deviation from the generator beta model is stronger when $n$ is large.   

In order to fit both regression models (bessel and beta) via EM algorithm, consider the initial values (default choices from \texttt{betareg}) previously described in Section \ref{sec:sim}. In terms of link functions, the standard options in \texttt{betareg} (logit for the mean and log for the precision) are also used here for both cases. In line with the bessel EM, the maximization step within the beta EM is also performed through \texttt{optim} with \texttt{BFGS}. The convergence criterion is again based on the tolerance $\epsilon = 10^{-5}$. We emphasize the fact that the analyses are developed by fitting both models to the same MC samples. The EM algorithms are executed independently for each sample using parallel computing via \texttt{snowfall}.

\begin{figure}[!h]
	$$
	\begin{array}{c}
	\includegraphics[scale=0.28]{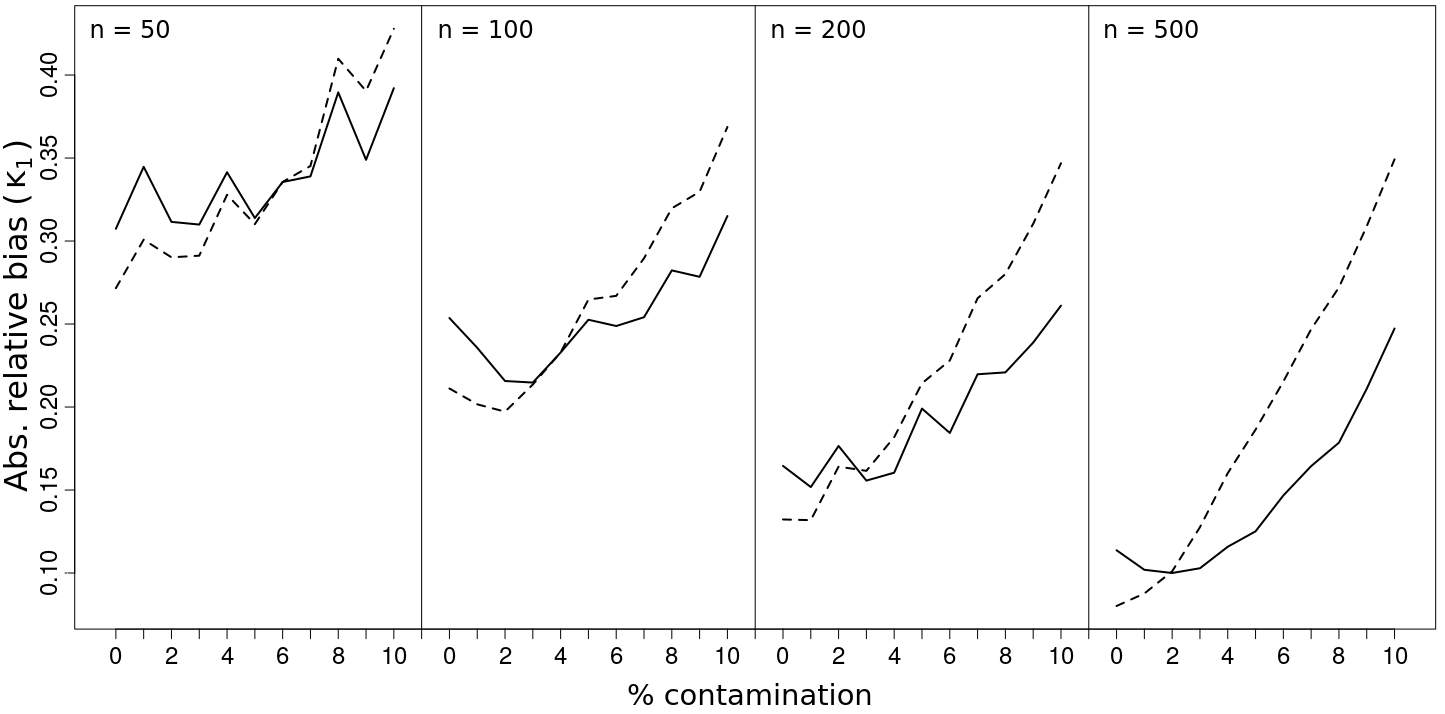} \\
	\includegraphics[scale=0.28]{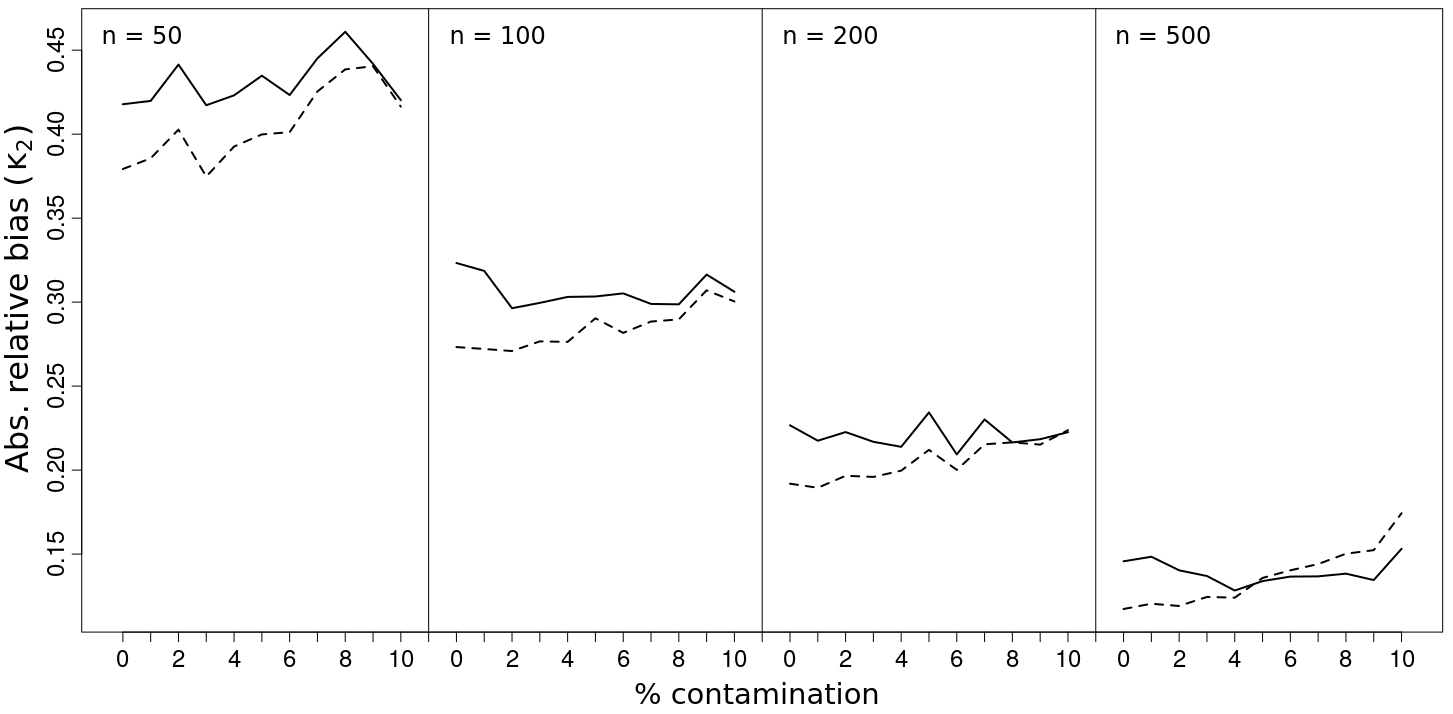} \\
	\includegraphics[scale=0.28]{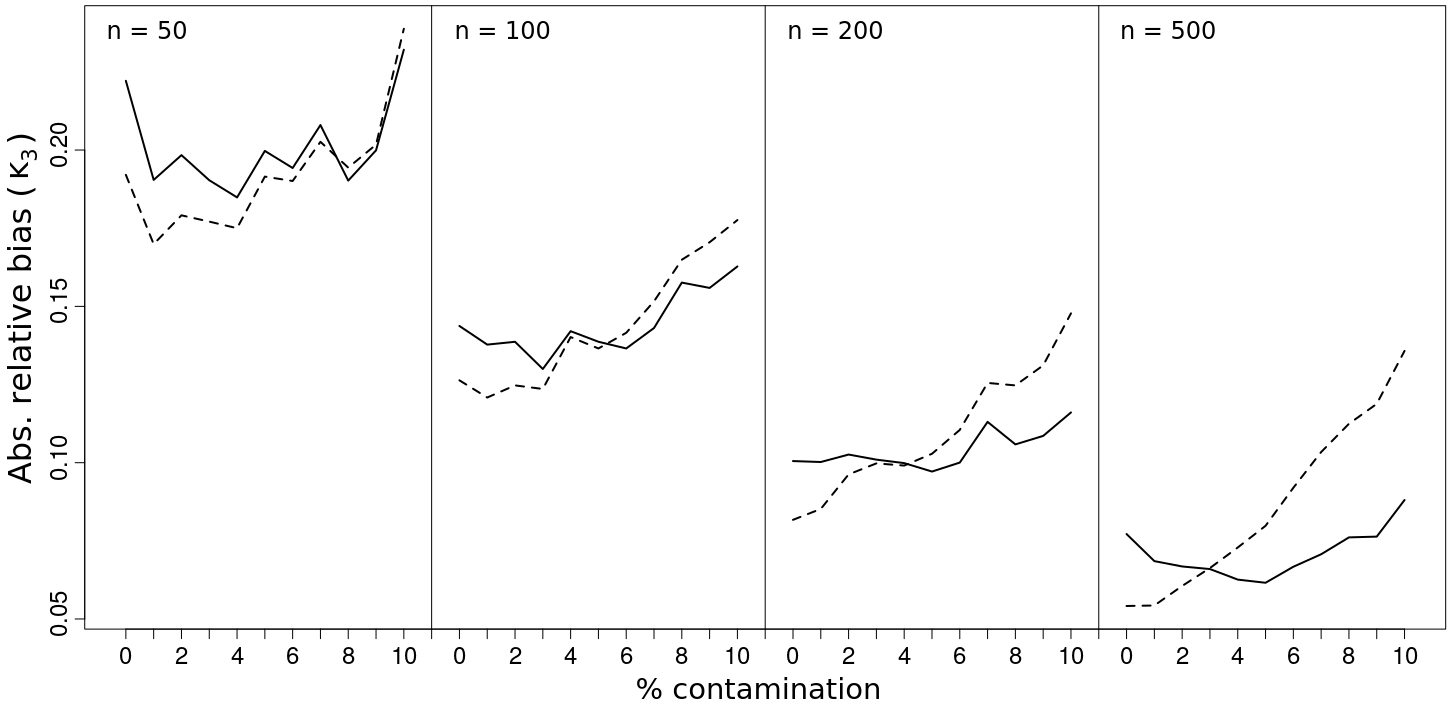} \\
	\end{array}
	$$  
	\caption{Absolute relative bias for the coefficients in ${\boldsymbol \kappa}$. Comparison between the bessel (solid line) and the beta (dashed line) regressions. In each panel, the graphs from left to right correspond to the sample sizes $n = 50$, $100$, $200$ and $500$. The curves are built with respect to different percentages of contaminations (ranging from $\approx 0\%$ to $\approx 10\%$).} \label{fig_rbias}
\end{figure}

Figure \ref{fig_rbias} compares the bessel (solid line) and beta (dashed line) regressions in terms of absolute relative bias for the coefficients in ${\boldsymbol \kappa}$. The comparison accounts for different percentages of contaminations (approximately $0\%$, $1\%$, $\cdots$, $10\%$) and different sample sizes. The absolute relative bias of an estimate $\hat{\kappa}_j$ is given by $|(\hat{\kappa}_j-\kappa_j)/\kappa_j|$, where $\kappa_j$ is the true value used to generate the data. This quantity is calculated for each MC replication and then the average is taken as the final outcome. Note that this statistic is essentially a ratio between the estimation error and the true value of the parameter, hence a large error (numerator) and a true value near zero (denominator) will provide a high absolute relative bias. If the true value is large, the magnitude of the error (numerator) must increase to maintain the same level of relative bias. This denominator can be seen as a penalty in the analysis of the bias for any parameter having a true value close to zero. This idea is reasonable in the context of regression analysis, since a coefficient near zero can be easily regarded as not significant. The discussion of Figure \ref{fig_rbias} is focused on the absolute relative bias, however, similar conclusions can be obtained when exploring the mean square error (not reported in this paper).  

The results displayed in Figure \ref{fig_rbias} show that the largest absolute relative bias is lower than $0.50$; see the panel for $\kappa_2$ and $n = 50$. As expected, the bias decreases as the sample size increases; this can be noted through the decreasing height of the curves from left to right in each panel. In addition, most curves have in general an increasing behavior suggesting that bias increases with the percentage of contamination. Now focusing on the intercept $\kappa_1$ (first panel), it is clear from the visual inspection that the bias related to the bessel regression (solid curve) is lower than that from the beta regression (dashed curve) for large percentages of contamination. The same conclusion can be drawn in the analysis of $\kappa_3$ in the third panel. The beta regression naturally provides better results in the scenario without ($0\%$) or having low contamination, since the data is generated under the beta model itself. Another aspect to be highlighted, looking at $10\%$ from the first and third panels, is the fact that the difference between the solid and dashed curves seems to increase with $n$. As previously discussed, the number of contaminated values in the sample depends on $n$, therefore, the bessel regression does a better job when $n$ is large; the case $n = 50$ ($\kappa_1$ and $\kappa_3$) shows that both models have a similar performance for large contaminations. Recall that the coefficient $\kappa_2$ is related to a binary covariate and the effect of this type of regressor is known to be harder to estimate. This point can explain the results in the second panel, where the solid curve seems flat for all $n$ and the dashed curve has a slow increase. Note that the bessel regression indicates lower bias in $\kappa_2$ for $n = 500$ and more than $6\%$ of contamination.   

Results reported in this section have confirmed the robustness and flexibility of the bessel regression to handle data originated from a different generator model (beta). The bessel model have shown to be a strong competitor to the beta regression under misspecification caused by the presence of contaminations in the sample. In the next section, we discuss results from three real applications. The first two cases involves data sets tagged as ``bessel regression'' through the DBB procedure. The third application is based on a data set detected as ``beta regression'' by the DBB test.

\section{Empirical illustrations}\label{sec:emp}

This section presents three real applications for which both bessel and beta regressions are fitted and evaluated. The main aim is to compare model performances, explore residuals and the predictive accuracy.

\subsection{Stress/anxiety data}\label{sec:data1}

The first application is based on a data set refering to a study involving $166$ women in Australia. The data is available through the \texttt{R} package \texttt{betareg}; see also the reference \cite{sv2006} for details. There are two variables for the analysis: the response variable is denoted by \texttt{stress} and the covariate is called \texttt{anxiety}. These values were originally measured in a depression anxiety stress scale, being scores ranging from 0 to 42. \cite{sv2006} applied a linear transformation to rescale them to the unit interval.

The DBB discrimination test (Section \ref{sec:discrimination}) determines the bessel regression as the most appropriate model for this application. The main results in the structure of the test are: $\sum_{i=1}^{166} z_i^2/166 = 0.02577$, $\sum_{i=1}^{166} [\frac{1}{2} \widetilde{\mu}_i (1-\widetilde{\mu}_i) + \widetilde{\mu}_i^2] = 9.11992$, $|D_{bessel}| = 0.001050$ and $|D_{beta}| = 0.00211$. Table \ref{tab:sa} presents the estimated coefficients from bessel and beta regressions. The plan here is to simplify the analysis, therefore, covariates explaining $\phi_i$ are not included in this application. As a consequence of this particular choice, $\phi_i = \exp\{\lambda_1\}$ is constant for all $i$.    

\begin{table}[!h]
	\centering
	\tabcolsep=2.5pt
	\begin{tabular}{c|crr}
		\hline
		Parameter & Covariate & \multicolumn{1}{c}{bessel} & \multicolumn{1}{c}{beta} \\
		\hline
		$\kappa_1$  & intercept & $-3.298$ ($0.139$) & $-3.480$ ($0.143$) \\
		$\kappa_2$  & anxiety & $3.200$ ($0.336$) & $3.752$ ($0.316$) \\
		\hline 
		$\lambda_1$ & intercept & $1.543$ ($0.204$) & $2.458$ ($0.123$) \\
		$g(\phi_i)$ & - & \multicolumn{1}{c}{$0.136$} & \multicolumn{1}{c}{$0.079$} \\ 
		\hline
	\end{tabular}
	\caption{Comparison between bessel and beta regression models for the stress/anxiety data. Estimates of the coefficients in ${\boldsymbol \kappa}$ and $\lambda_1$; standard errors are in parentheses. Estimates of $g(\phi_i)$, based on the intercept $\lambda_1$, are given at the bottom.} \label{tab:sa}
\end{table}

When confronting both models, note that the estimates of ${\boldsymbol \kappa_2}$ (related to the covariate anxiety) are not close. The coefficient of ``anxiety'' is positive for both models, suggesting that an increase in the anxiety score leads to a higher stress level. The two models are not comparable in terms of $\phi_i$; the confrontation must be done using $g(\phi_i)$, defined in $Var(Z_i) = \mu_i (1-\mu_i) g(\phi_i)$. The value of $g(\phi_i)$, see Table \ref{tab:sa}, is larger for the bessel regression (almost the double of the beta one).

Figure \ref{fig_env0} presents Pearson residuals and simulated envelopes plotted against the quantiles of the standard normal distribution. This distribution serves as a basis to build the envelopes and any other choice can be considered for the same purpose. The Pearson residual is widely used to explore generalized linear models and has the advantage of accounting for both the mean and the variability. The distinction in variability between beta and bessel is a key motivation for applying the Pearson residuals in our study. The comparative analysis would be unfair if this point is ignored. According to \cite{bss17}, these residuals are expected to be concentrated around zero and the $N(0,1)$ quantiles are usually considered for comparison due to asymptotic properties. These authors also emphasize that the Gaussian approximation can be poor for small or moderate sample sizes. The Pearson residuals are defined as:
\begin{equation}
R_i = (Z_i - \widehat\mu_i)/\sqrt{\widehat\mu_i (1-\widehat\mu_i) g(\widehat\phi_i)}, \label{residP}    
\end{equation}
with $g(\phi_i)$ given in (\ref{gphi_bet}) and (\ref{gphi_bes}) for the beta and bessel case, respectively. In order to build the simulated envelopes, the following algorithm is considered:
\begin{enumerate}
	\item Choose a regression setting (bessel or beta).
	\item Use the EM algorithm to fit the real data and save $\widehat\mu_i$ and $g(\widehat\phi_i)$ for $i = 1, \cdots, n$. 
	\item Generate $1000$ synthetic data sets. Simulate the responses $Z_i$'s using the same covariates, $\mu_i$ and $g(\phi_i)$ from step 2.
	\item Fit the chosen model to the $1000$ artificial data sets. Save $\widehat\mu_i^{(j)}$ and $g(\widehat\phi_i)^{(j)}$ for each data set $j$. 
	\item Compute the Pearson residuals $R_i^{(j)}$ in (\ref{residP}) for $j = 1, \cdots, 1000$. Use $\widehat\mu_i^{(j)}$ and $g(\widehat\phi_i)^{(j)}$ from step 4. 
	\item Let $R$ be a $1000 \times n$ matrix with $R^{(j)} = \{ R_1^{(j)}, R_2^{(j)}, \cdots, R_n^{(j)} \}$ in the $j$-th row.
	\item Sort the rows of $R$ (ascending order).
	\item Sort the columns of $R$ (ascending order).
	\item Assuming a $95\%$ coverage, the lower and upper bounds of the target envelope are given by the $25^{th}$ and $975^{th}$ rows of $R$, respectively. 
\end{enumerate}

In Figure \ref{fig_env0}, it is quite clear that the envelopes computed via bessel regression incorporate almost all points representing the residuals from the real data. The beta regression does not have the same performance; several points in the region between 0 and 1 (horizontal axis) are positioned outside the range of the envelope. The percentages of points within the shaded region are: $91.57\%$ for the bessel and $59.04\%$ for the beta. This result is in accordance with the DBB test and visually indicates how better is the performance of the bessel regression with respect to the beta model in the current application. 

\begin{figure}[!h]
	\centering
	\includegraphics[scale=0.30]{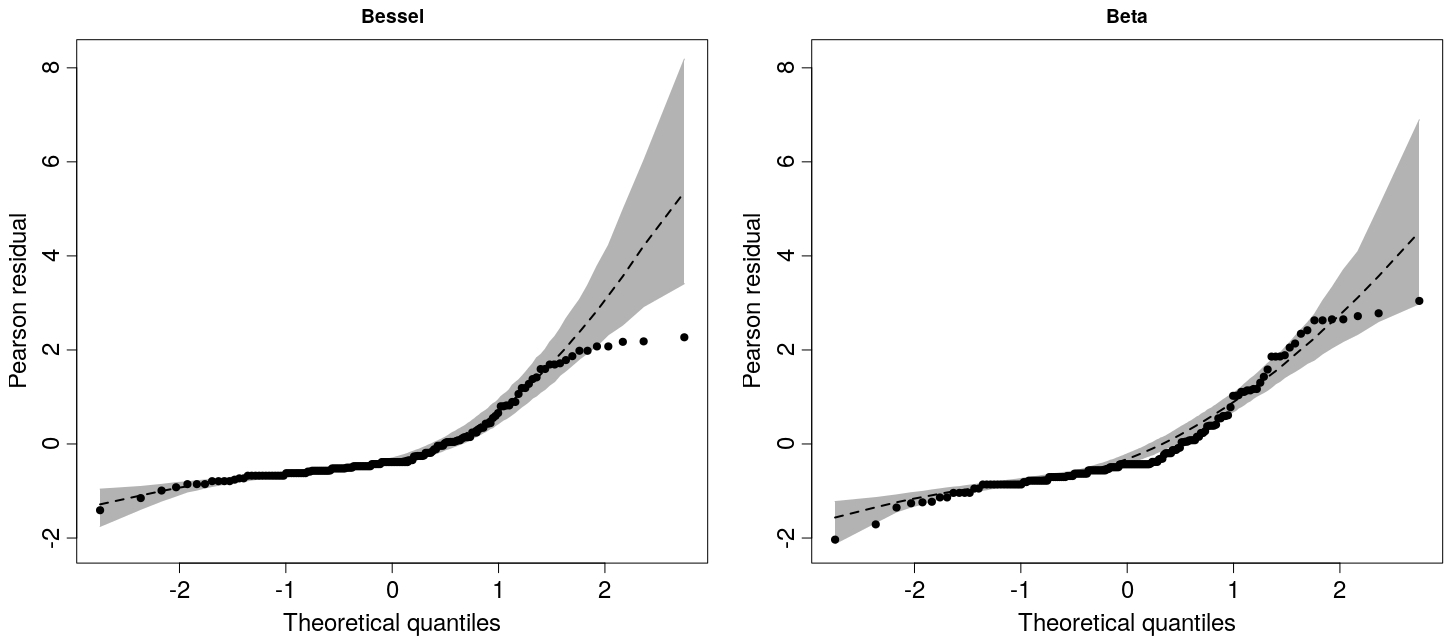}
	\caption{Pearson residuals against the theoretical quantiles from the standard normal distribution. The points are the residuals for the stress/anxiety data set. The shaded area represents a $95\%$ envelope based on $1000$ simulations from each regression. The dashed line indicates the envelope mean.} \label{fig_env0}
\end{figure}

\begin{figure}[!h]
	\centering
	$$
	\begin{array}{cc}
	(a) \hspace{3.5cm} (b) & (c) \hspace{3.5cm} (d) \\
	\includegraphics[scale=0.22]{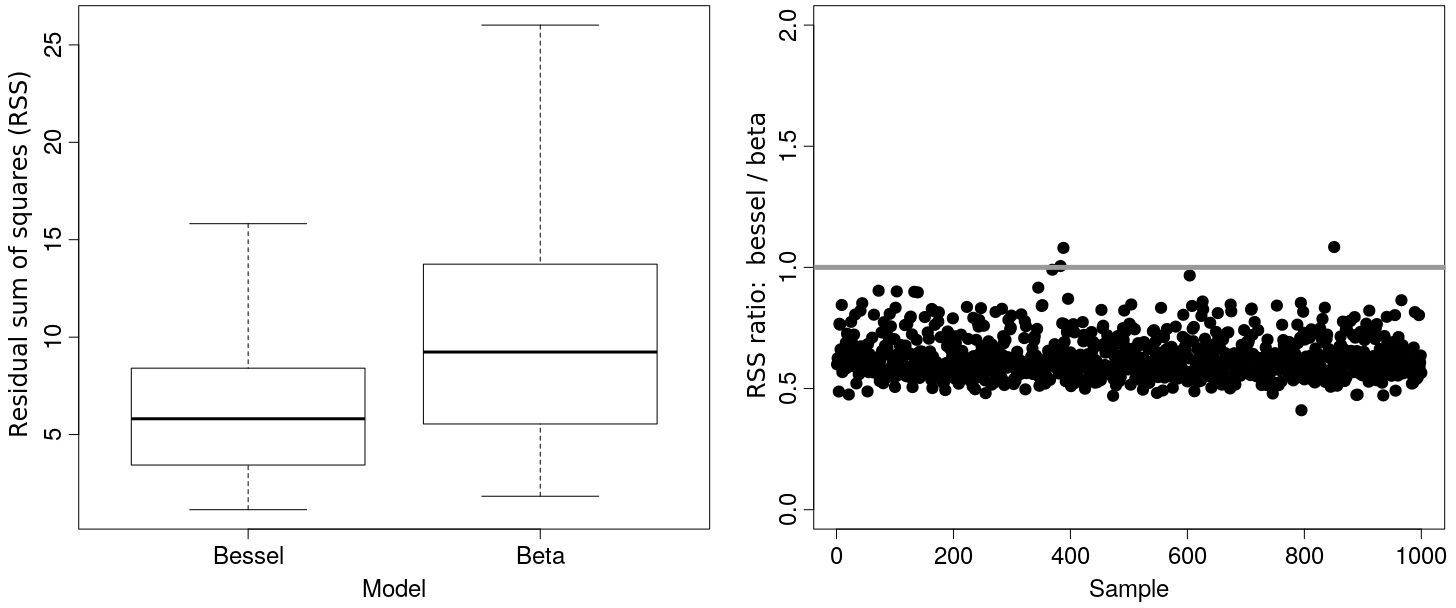} &
	\includegraphics[scale=0.22]{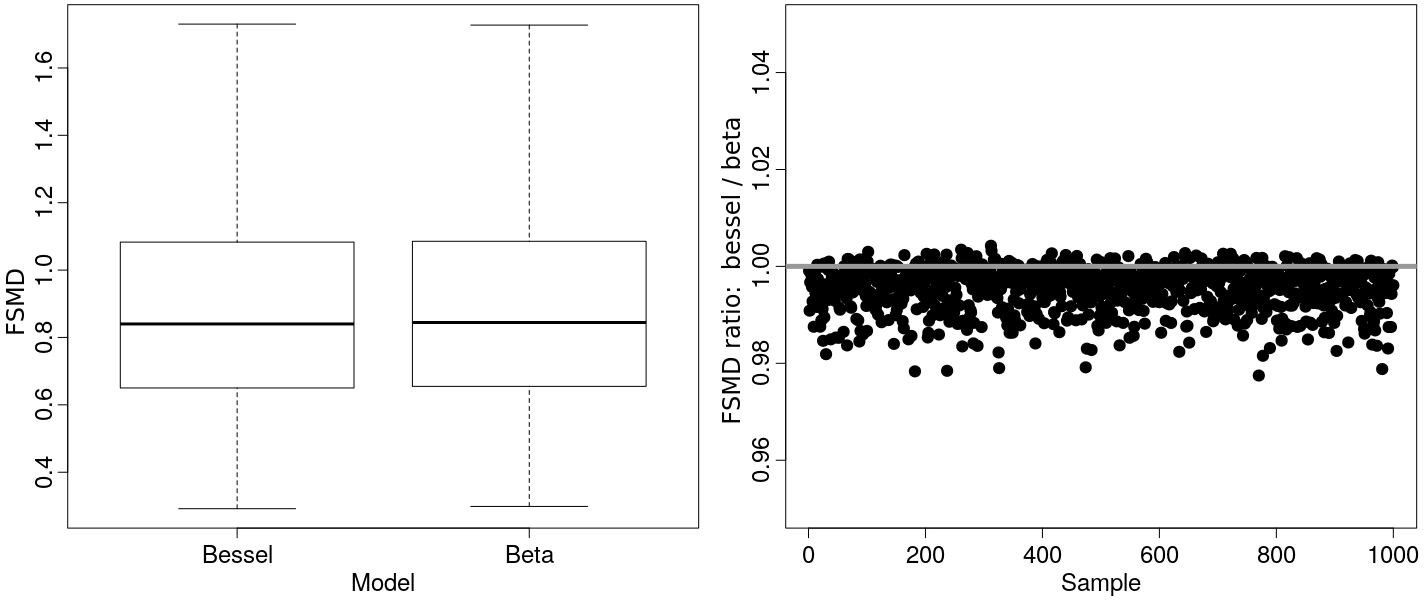} \\
	\end{array}
	$$
	\vspace{-15pt}
	\caption{Predictive accuracy of the regression models. Cross validation study based on partitions ($1000$ partitions having a training and a test set) of the stress/anxiety data set. Pearson residual and FSMD statistic are computed for $10$ randomly selected observations forming the test set. Panels $(a)$ and $(c)$ show boxplots of the RSS and the FSMD for each model, respectively. Panels $(b)$ and $(d)$ indicate the ratio of RSS's and FSMD's (bessel over beta) for each partition, respectively.} \label{fig_rss0}
\end{figure}

Another important feature to be explored when comparing different models is the predictive accuracy. Some authors have discussed about distinctions when confronting models in terms of goodness-of-fit and predictive performance. This is a central concern in the machine learning field; see \cite{loyer16}, \cite{yark17} and references therein for example. In brief, goodness-of-fit is how well a model can explain or accommodate all data points used to estimate the parameters; whereas, predictive accuracy represents how well a model can approximate new data points, which were not used to fit the model in the first place. The model showing the best goodness-of-fit result may not be the most accurate in terms of prediction, and vice versa. As a result of this fact, there is a choice to be made here and some researchers tend to prefer the prediction accuracy over goodness-of-fit in their model selection criterion. 

Figure \ref{fig_rss0} shows four panels comparing the behavior of the residual sum of squares (RSS) and the statistic FSMD (first and second moments distances) for both models. The FSMD statistic is an alternative formulation accounting for the separations in terms of mean and variance. The FSMD based on a sample of size $n$ is given by 
\begin{equation}
 \mbox{FSMD} \ = \ \sum_{i=1}^{n} S_i, \quad \mbox{with} \quad S_i = |Z_i - \widehat E(Z_i)| + |Z_i^2 - \widehat E(Z_i^2)|. \label{fsmd}
\end{equation}
Note that a small FSMD indicates better performance. 

The analysis in Figure \ref{fig_rss0} is designed to explore the regressions in terms of prediction. The steps to compute the RSS and FSMD are:
\begin{enumerate}
	\item Choose a regression setting (bessel or beta) and let $j = 1$.
	\item Separate the full real data set in two parts: 10 observations are randomly selected to form the ``test set'' and the remaining ones are considered in the ``training set''. 
	\item Use the training set to fit the regression model and estimate parameters.
	\item Estimate $\mu_i$ and $g(\phi_i)$ using the covariates related to the observations in the test set.
	\item Compute the Pearson residual $R_i$ in (\ref{residP}) and the term $S_i$ in (\ref{fsmd}), with $Z_i$ being an observation in the test set. 
	\item Calculate \ $\mbox{RRS}^{(j)} = \sum_{i=1}^{10} R_i^2$ \ and \ $\mbox{FSMD}^{(j)} = \sum_{i=1}^{10} S_i$.
	\item If $j < 1000$, update the iteration to be $j+1$ and return to step 2. 
\end{enumerate}

Note that the RSS or FSMD calculation is based on $1000$ random partitions of the $166$ observations in the stress/anxiety data. Both models are fitted to each partition. Figure \ref{fig_rss0} $(a)$ compares the two model in terms of boxplots for the RSS's. As can be seen, the boxplot associated to the bessel model is lower than the one for the beta case. Panel $(c)$ indicates some similarity between the models in terms of boxplots for the FSMD. Panel $(b)$ shows points representing the ratio $\mbox{RSS}^{(j)}_{\tiny \mbox{bessel}}/\mbox{RSS}^{(j)}_{\tiny \mbox{beta}}$ for each partition $j = 1, \cdots, 1000$. The horizontal grey line in the graph identifies the ratio $1$ representing the scenario where the RSS's are the same for both models. Note that almost all points are located below the grey line, indicating that $\mbox{RSS}_{\tiny \mbox{bessel}} < \mbox{RSS}_{\tiny \mbox{beta}}$ for all partitions. Panel $(d)$ reinforces the conclusion from $(b)$ with $\mbox{FSMD}_{\tiny \mbox{bessel}} < \mbox{FSMD}_{\tiny \mbox{beta}}$ for most cases. In summary, the results presented in Figure \ref{fig_rss0} clearly suggest that the bessel regression has a better predictive performance than the beta model.

One may argue about using other residuals specifically proposed in the literature to enhance the ability of assessing the goodness-of-fit in a beta regression. This must be considered with caution, since an appropriate residual for the beta case may not be suitable for other models (misleading any comparative analysis). An interesting option is the quantile residual designed for the beta setting and evaluated in \cite{pereira19}. Figure \ref{fig_envq} shows a remarkable result for the bessel in terms of residual envelopes. The quantile residuals are better accomodated by the bessel regression (86.75\% of the points are captured by the simulated envelopes, whereas this percentage is 59.04\% in the beta case).

\begin{figure}[!h]
	\centering
	\includegraphics[scale=0.30]{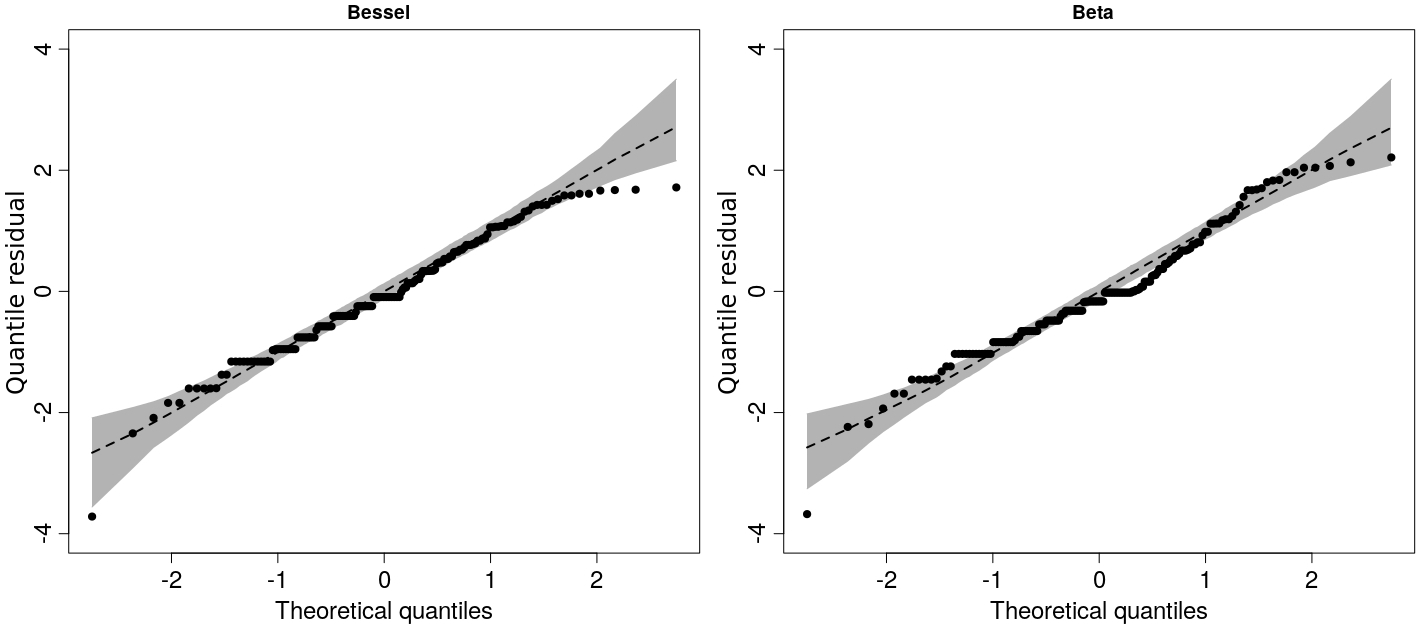}
	\caption{Quantile residuals \citep{pereira19} against the theoretical quantiles from the standard normal distribution. The points are the residuals for the stress/anxiety data set. The shaded area represents a $95\%$ envelope based on $1000$ simulations from each regression. The dashed line indicates the envelope mean.} \label{fig_envq}
\end{figure}

\subsection{Weather task data}\label{sec:data2}

In this second application, we investigate the freely available data set labelled as \texttt{WeatherTask} in the \texttt{R} package \texttt{betareg}; it is also referenced in \cite{Smith09} and \cite{Smith11}. The data correspond to a study where 345 participants were requested to evaluate how likely Sunday is to be the hottest day of the next week. All individuals were either first or second-year undergraduate students in psychology with weak background in probability theory. The dependent variable, denominated \texttt{agreement}, is the probability or the average between two probabilities indicated by each individual. Two covariates are considered in the analysis. The first one \texttt{priming} has two categories: ``two-fold'' and ``seven-fold''. The two-fold case is related to the question ``what is the probability that the temperature on Sunday will be higher than any other day next week?'', which induces the partition ``Sunday hotter'' or ``Sunday not hotter'' priming the individual ignorance prior in two parts. In contrast, the seven-fold case is related to the question ``what is the probability that the highest temperature of next week will occur on Sunday?'', which induces a partition with seven components (Sunday hottest, Monday hottest, Tuesday hottest, etc). The second covariate \texttt{eliciting} is also categorical with two scenarios: ``precise'' (the student is requested to provide a single probability as response) and ``imprecise'' (the student is required to assign lower and upper probabilities). Both covariates are treated as binary with 1 representing ``seven-fold'' and ``imprecise'', respectively. 

The DBB discrimination test (Section \ref{sec:discrimination}) indicates the bessel model as the most appropriate option for this case. The main results supporting this conclusion are: $\sum_{i=1}^{345} z_i^2/345 = 0.08525$, $\sum_{i=1}^{345} [\frac{1}{2} \widetilde{\mu}_i (1-\widetilde{\mu}_i) + \widetilde{\mu}_i^2] = 54.62012$, $|D_{bessel}| = 0.00039$ and $|D_{beta}| = 0.00296$. Table \ref{tab:wt} shows the estimates of the coefficients from both models. We do not include covariates to explain $\phi_i$ in this application for simplicity; therefore, $\phi_i = \exp\{\lambda_1\}$ is constant for all $i$.    

\begin{table}[!h]
	\centering
	\tabcolsep=2.5pt
	\begin{tabular}{c|crr}
		\hline
		Parameter & Covariate & \multicolumn{1}{c}{bessel} & \multicolumn{1}{c}{beta} \\
		\hline
		$\kappa_1$  & intercept & $-$1.154 (0.071) & $-$1.135 (0.071) \\
		$\kappa_2$  & priming & $-$0.255 (0.079) & $-$0.300 (0.081) \\
		$\kappa_3$  & eliciting & 0.339 (0.079) & 0.331 (0.081) \\
		\hline 
		$\lambda_1$ & intercept & 1.595 (0.097) & 2.036 (0.074) \\
		$g(\phi_i)$ & - & \multicolumn{1}{c}{0.132} & \multicolumn{1}{c}{0.116} \\ 
		\hline
	\end{tabular}
	\caption{Comparison between bessel and beta regression models for the weather task data. Estimates of the coefficients in ${\boldsymbol \kappa}$ and $\lambda_1$; standard errors are in parentheses. Estimates of $g(\phi_i)$, based on the intercept $\lambda_1$, are given at the bottom.} \label{tab:wt}
\end{table}

The estimates of ${\boldsymbol \kappa}$ and their standard errors are quite equivalent when comparing both models. The coefficient of ``priming'' is negative, suggesting that students in the seven-fold category tend to respond smaller probabilities than those in the two-fold case. The covariate ``eliciting'' has a positive coefficient, indicating that students in the imprecise category tend to inflate the probability as compared to those in the ``precise'' group. These regressions cannot be compared in terms of $\phi_i$ itself, but they can be confronted with respect to $g(\phi_i)$ appearing in $Var(Z_i) = \mu_i (1-\mu_i) g(\phi_i)$. Note that the values of $g(\phi_i)$ in Table \ref{tab:wt} are near, suggesting a similar performance.

Figure \ref{fig_env1} displays the Pearson residuals and simulated envelopes against the quantiles of the standard normal distribution. Revisit Section \ref{sec:data1} for details about how the envelopes are built. Note that the points are better accommodated under the envelope based on the bessel model. In particular, we highlight the group of points related to the interval between 2 and 3 in the horizontal axis. These points are clearly located within the bessel envelope and they lie outside the range of the beta envelope. Once again, we can also say here that this envelope analysis agrees with the DBB test; the visual inspection of the residual graph indicates how better is the performance of the bessel with respect to the beta model in this second empirical illustration. 

\begin{figure}[!h]
	\centering
	\includegraphics[scale=0.30]{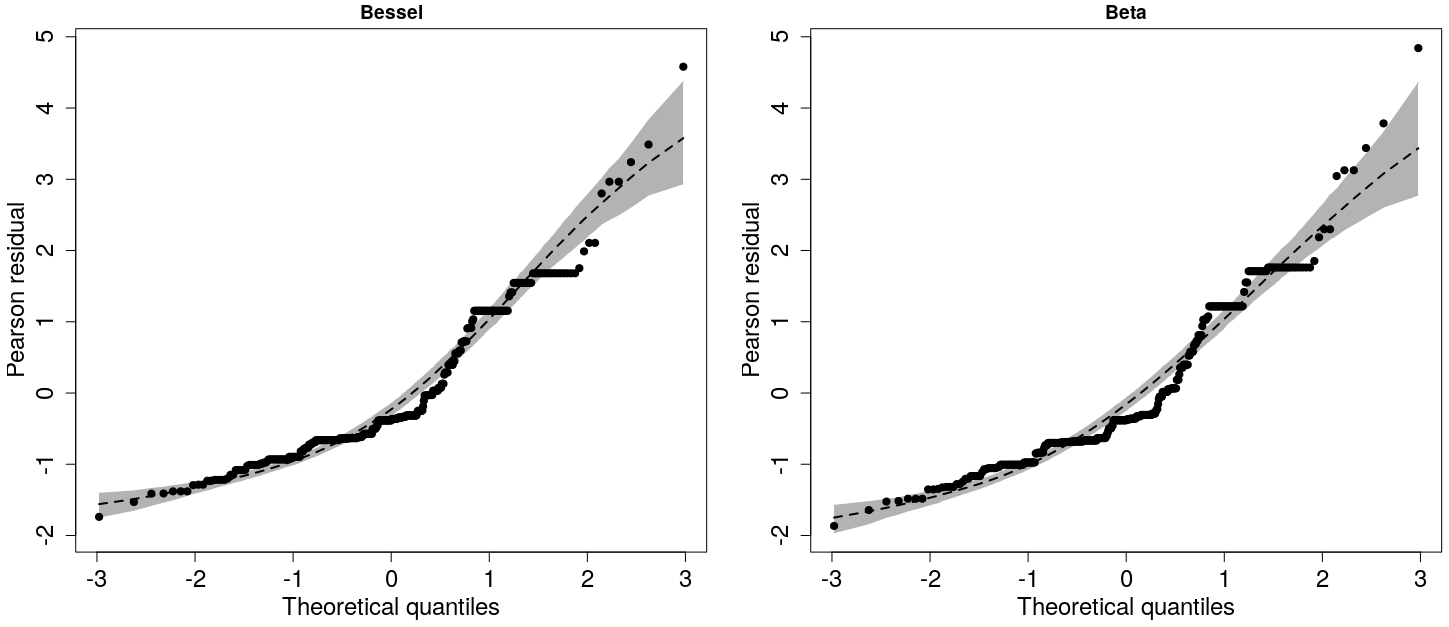}
	\caption{Pearson residuals against the theoretical quantiles from the standard normal distribution. The points are the residuals for the Weather task data set. The shaded area represents a $95\%$ envelope  based on $1000$ simulations from each regression. The dashed line indicates the envelope mean.} \label{fig_env1}
\end{figure}

Figure \ref{fig_rss1} compares the residual sum of squares and FSMD statistic for both regressions. The RSS and FSMD are obtained based on $1000$ random partitions of the $345$ observations in the weather task data; review the steps described in Section \ref{sec:data1} to compute the RSS's and FSMD's. The two models are fitted to each partition. Figure \ref{fig_rss1} $(a)$ compares the models with respect to their RSS's. Is is possible to see again that the bessel boxplot is slightly lower than the beta one. The boxplots in Panel $(c)$ suggest a similarity between the FSMD of the models in this illustration. Panel $(b)$ shows the ratios $\mbox{RSS}^{(j)}_{\tiny \mbox{bessel}}/\mbox{RSS}^{(j)}_{\tiny \mbox{beta}}$ for each partition $j = 1, \cdots, 1000$. The horizontal grey line represents the ratio $1$ (equal RSS's). In this case, all points are located below $1$, indicating $\mbox{RSS}_{\tiny \mbox{bessel}} < \mbox{RSS}_{\tiny \mbox{beta}}$ for all partitions. Panel $(d)$ provides a similar conclusion with $\mbox{FSMD}_{\tiny \mbox{bessel}} < \mbox{FSMD}_{\tiny \mbox{beta}}$ for most partitions. In conclusion, these results from Figure \ref{fig_rss1} strongly suggest that the bessel regression has again a better predictive performance than the beta model.

\begin{figure}[!h]
	\centering
	$$
	\begin{array}{cc}
	(a) \hspace{3.5cm} (b) & (c) \hspace{3.5cm} (d) \\
	\includegraphics[scale=0.22]{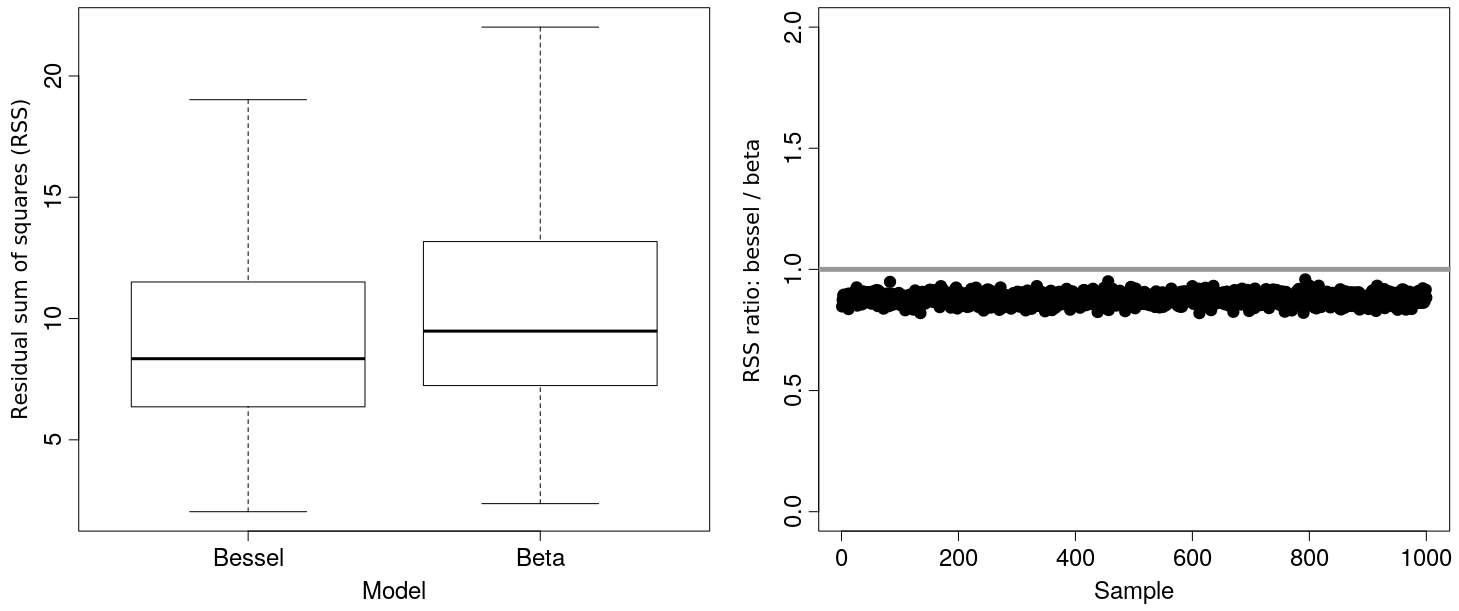} & 
	\includegraphics[scale=0.22]{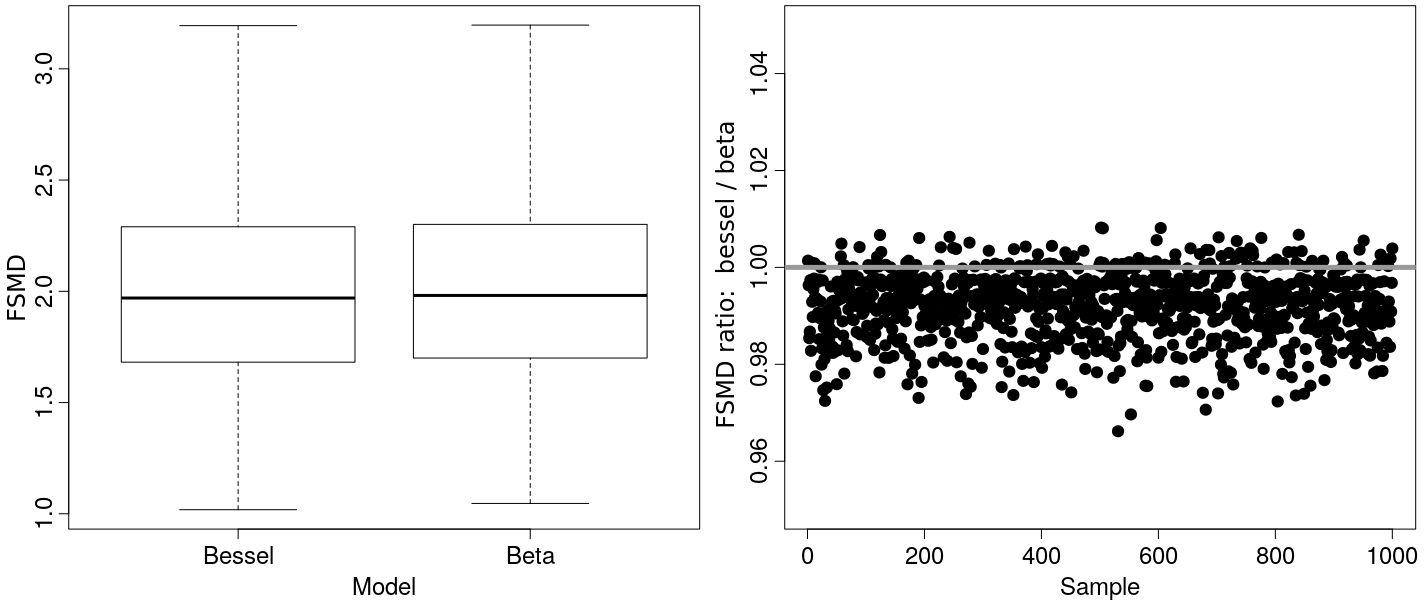} \\
	\end{array}
	$$
	\vspace{-15pt}	
	\caption{Predictive accuracy of the regression models. Cross validation study based on partitions ($1000$ partitions having a training and a test set) of the weather task data set. Pearson residual and FSMD statistic are computed for $10$ randomly selected observations forming the test set. Panels $(a)$ and $(c)$ show boxplots of the RSS and the FSMD for each model, respectively. Panels $(b)$ and $(d)$ indicate the ratio of RSS's and FSMD's (bessel over beta) for each partition, respectively.} \label{fig_rss1}
\end{figure}

\begin{figure}[!h]
	\centering
	\includegraphics[scale=0.35]{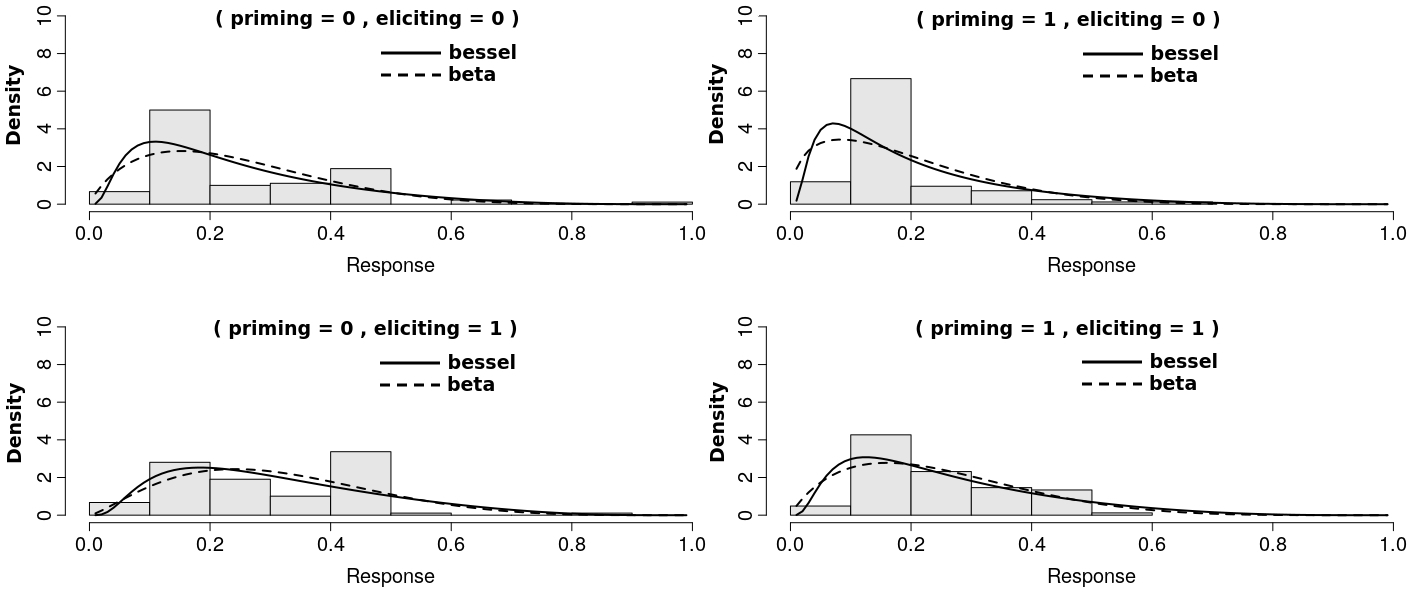}
	\caption{Data histograms for each combination of the binary covariates (priming, eliciting) in the Weather task data. Densities of the corresponding bessel and beta models overlay each graph.} \label{fig_hist_pdf}
\end{figure}

The application discussed in this section involves two categorical covariates, providing a convenient scenario to show the estimated density against an empirical distribution. The presence of continuous covariates, in other illustrations of this paper, determines specific $\mu_i$ and $\phi_i$ for each element $i$. This makes this type of visual inspection infeasible for those cases. The Weather task data set contains: 90 observations with (\texttt{priming} = 0, \texttt{eliciting} = 0), 89 cases with (0, 1), 84 with (1, 0) and 82 with (1, 1). Figure \ref{fig_hist_pdf} shows the histograms for each combination. The graphs are overlaid by the bessel and beta densities, which are obtained by fixing the covariates in the mentioned pair of values. The parameter estimation is based on the whole data set. Note that the response tend to be concentrated below 0.5 and the shapes of both densities agree with this behavior. Few differences can be seen between the dashed and continuous curves. It is difficult to judge the models based on this visual analysis. The other criteria, previously discussed, should be regarded for the model comparison.

\subsection{Body fat data}\label{sec:data3}

This section is dedicated to a real data analysis involving measurements of body fat as the response variable. The percentage of body fat is an important variable to be considered when evaluating the health of an individual; however, measuring this quantity with accuracy is not a simple task. A reliable strategy involves weighting the individual submersed in water, which is rarely done in practice. As an alternative to this difficulty, the percentage of body fat can be predicted from other body measurements considered easier to obtain. This is a central point motivating the use of a regression model to study body fat. The main goal of the present section is to explore the bessel model in this real application and then compare its results with those from the beta regression.

The data set explored here is freely available online and it is usually known as the Penrose body fat study \citep{pen85}. This data set is composed by several physiologic measurements related to 252 men. Their percentage of body fat (response variable) was obtained via the underwater weighting technique. In terms of covariates, the following options are included: age (years), weight (lbs), height (inches) and the circumferences (in cm) of the neck, chest, abdomen, hip, thigh, knee, ankle, biceps, forearm and wrist. In order to avoid computational issues related to the magnitude of these covariates, they are rescaled dividing their original value by $100$. Note that this modification does not imply in any loss for the analysis; the estimated coefficients are now multiplied by $100$ compared to the original scale. As reported in \cite{bri16}, some of these regressors are highly correlated leading to the multicollinearity issue, therefore, variable selection is required. We choose to select the most important regressors for our application using the variance inflation factors VIF \citep{mon12} to investigate multicollinearity.     

In an exploratory descriptive analysis using the boxplot interquartile range, we have noted that subject $42$ is an outlier for the covariate height ($29.5$ inches = $74.93$ cm). This value is extremely inconsistent for a man being $44$ years old and weighting $\approx 93$ Kg; therefore, we decided to remove this individual from the study. Another outlier, found when inspecting the covariate weight, is the subject $39$ ($164.72$ Kg). This individual is also an outlier for almost all circumference measurements, which is an expected result. In order to evaluate the VIF, based on linear regressions assuming normality, subject $39$ is removed from this specific analysis to avoid deviation from the central assumption. All histograms for each covariate (without subjects $39$ and $42$) clearly resemble the shape of a Gaussian density. The reader is advised that subject 39 is ignored in the VIF analysis, but he will be reinserted in the data set for the comparison between the bessel and beta regressions. 

In the VIF analysis, we choose to be rigorous by demanding a value below 5 to refute the multicollinearity. This is reasonable for the present application containing several covariates under suspicion of being intercorrelated. The following steps are considered in our evaluation: ($i$) fit an ordinary least square regression having the covariate $X_j$ as a function of all the other 12 explanatory variables, for $j = 1, 2, \cdots, 13$; ($ii$) calculate $\mbox{VIF}_j = 1/(1-R^2_j)$, where $R^2_j$ is the multiple $R^2$ of the $j$-th model; ($iii$) identify the largest $\mbox{VIF}_j \geq 5$ and remove the corresponding $X_j$ from the data set; ($iv$) return to step 1, without $X_j$, and repeat the procedure until all VIF's are smaller than $5$. According to this selection approach for the body fat data, the covariates removed from the study are: weight, abdomen and hip. Note that the VIF investigation is entirely focused on the relationship among covariates. The next stage is to verify which of the remaining 10 regressors have significant impact over the percentage of body fat. This is done by fitting a beta regression via \texttt{betareg} (without subjects $39$ and $42$, assuming the logit link for the mean, using maximum likelihood and ignoring covariates for the precision parameter). The result indicates that the regressors height (p-value: 0.146), neck (0.368), knee (0.686), ankle (0.825), biceps (0.954) and forearm (0.697) are not significant; therefore, they are removed from the study comparing bessel and beta. The same conclusion is obtained, if one chooses to remove a single variable (highest p-value) and refit the model until all coefficients are significant. For each subject $i$, the final set of covariates contains: age ($x_{i2}$), chest ($x_{i3}$), thigh ($x_{i4}$) and wrist ($x_{i5}$).   

We now compare results, from bessel and beta regressions, for the discussed body fat data fitted with the four selected covariates. The configuration of initial values, convergence tolerance, link functions and maximization method described in Section \ref{sec:miss} are also adopted here. The full data set with $251$ observations, without subject $42$, is considered in the main investigation. There is no inconsistency justifying the removal of the atypical individual $39$; recall that this action was taken during the VIF study to avoid deviations from normality. In any case, it is important to evaluate whether this atypical point is influential for the target models, thus the analysis of $250$ observations (without subjects 39 and 42) is also explored. In the context of body fat data, we do not have any information suggesting that the precision parameter could be explained by the available set of covariates; therefore, we decided to fit the simpler model with $\phi_i$ constant for all $i$.

\begin{table}[!h]
	\centering 
	\tabcolsep=1.4pt
	\scriptsize
	\begin{tabular}{c|crrrr|crr}
		\hline 
		& \multicolumn{5}{c|}{no interaction} & \multicolumn{3}{c}{with interaction} \\ 
		\cline{2-9}
		Parameter & Covariate & \multicolumn{1}{c}{bessel} & \multicolumn{1}{c}{beta} &  \multicolumn{1}{c}{bessel$^*$} & \multicolumn{1}{c|}{beta$^*$} & Covariate & \multicolumn{1}{c}{bessel} & \multicolumn{1}{c}{beta} \\
		\hline
		$\kappa_1$  & intercept & $-$10.787 (0.849) &  $-$5.385 (0.506) & $-$11.057 (0.833) & $-$5.854 (0.508) & intercept & $-$11.474 (0.586) & $-$6.269 (0.381)\\
		$\kappa_2$  & age &  2.253 (0.449) &   1.640 (0.251) & 2.329 (0.442) & 1.730 (0.246) & age & 2.079 (0.432) & 1.405 (0.242) \\
		$\kappa_3$  & chest &  5.096 (0.869) &   3.527 (0.508) & 5.042 (0.850) & 3.465 (0.494) & chest & 4.966 (0.848) & 3.338 (0.498) \\
		$\kappa_4$  & thigh &  9.069 (1.488) &   4.661 (0.854) & 9.552 (1.451) & 5.483 (0.853) & thigh & 9.012 (1.481) & 4.559 (0.852) \\
		$\kappa_5$  & wrist & $-$12.457 (6.955) & $-$17.443 (3.890) & $-$12.532 (6.832) & $-$17.437 (3.789) & wrist $\times$ height & $-$10.461 (5.866) & $-$15.125 (3.400) \\
		\hline
		$\lambda_1$ & intercept & 2.182 (0.124) &   3.616 (0.089) &  2.259 (0.123) & 3.669 (0.089) & intercept & 2.184 (0.124) & 3.615 (0.089) \\
		$g(\phi_i)$ & - & \multicolumn{1}{c}{0.086} & \multicolumn{1}{c}{0.026} & \multicolumn{1}{c}{0.081} & \multicolumn{1}{c}{0.025} & - & \multicolumn{1}{c}{0.086} & \multicolumn{1}{c}{0.026} \\
		\hline
	\end{tabular}
	\caption{Comparison between bessel and beta regressions for the Penrose body fat data. Estimates of the coefficients in ${\boldsymbol \kappa}$ and $\lambda_1$; standard errors are in parentheses. The table is divided in two parts: ($i$) model fit assuming wrist as the fourth covariate and ($ii$) model fit replacing wrist by the multiplicative interaction wrist $\times$ height. The cases marked with $*$ correspond to the analysis without the atypical subject $39$. Estimates of $g(\phi_i)$, based on the intercept $\lambda_1$, are given at the bottom.} \label{tab:bf}
\end{table}

Table \ref{tab:bf} shows the estimates of coefficients in three different scenarios. The first case (columns bessel and beta with no interaction) is the main model based on the four selected covariates from the exploratory analysis (subject 39 included). As can be seen, both regressions indicate the same sign for each coefficient. Note that age, chest and thigh have a positive impact over the mean $\mu_i$ of body fat percentage; their increase implies in $\mu_i$ increasing as well, which is reasonable. The wrist circumference has a negative impact over $\mu_i$ (large wrist connected with small body fat). In a fitness check, it is common to consider the size of the wrist connected with the person body frame. In fact, body frame size is usually evaluated using the wrist circumference in relation to the height; small wrist and small height determines the small-boned category for the body frame. In our pre-analysis, the height was removed due to its association with other regressors. In particular, the linear correlation between height and wrist is $0.397$ and the Pearson test indicates a p-value $\approx 0$, suggesting a significant linear association between them. Given this relationship, we decided to include in Table \ref{tab:bf} the results replacing the covariate wrist by the multiplicative interaction wrist $\times$ height (see the last three columns). The coefficient $\kappa_5$ is again negative, allowing us to conclude that when moving towards the large-boned frame (increasing wrist and height) a decrease is expected in the mean body fat. Although the sign (and interpretation) of each coefficient is similar for both models in this application, the reader should note that the estimates from bessel and beta are quite different in terms of magnitude; for instance, the intercept $\kappa_1$ under the bessel model is approximately twice the value under the beta regression. Another aspect worth noting, is the size of the standard errors. The bessel regression provides larger standard errors for all scenarios and coefficients. In any model, the coefficient showing the largest standard error is wrist or the interaction wrist $\times$ height. The third scenario explored in Table \ref{tab:bf} is a model fit without the atypical subject $39$ (no interaction, models marked with $*$). In general, results does not change much when removing the outlier, which indicates that this observation is not an influential point. The largest difference between estimates is observed for the covariate thigh under the beta model. Estimates for $g(\phi_i)$ are given at the bottom of Table \ref{tab:bf}. These estimates do not change significantly within the same type of model (beta or bessel). The smallest values are observed for the beta regression in this case.

The DBB test (see Section \ref{sec:discrimination}) indicates the beta regression as the suitable option for this application. This conclusion is based on the following results: $\sum_{i=1}^{250} z_i^2/250 = 0.04339$, $\sum_{i=1}^{250} [\frac{1}{2} \widetilde{\mu}_i (1-\widetilde{\mu}_i) + \widetilde{\mu}_i^2] = 29.08093$, $|D_{bessel}| = 0.02025$ and $|D_{beta}| = 0.00141$. The graphs of simulated envelopes against quantiles of $N(0,1)$ in Figure \ref{fig_env2} are in conformity with the DSS result. One can easily see that the envelope for the beta regression (right panel) better accommodates the Pearson residuals for the body fat data than the bessel case. The atypical observation related to the individual $39$ is the one showing the smallest residual; see the point located near the lower left corner of the graphs. The envelopes in both panels do not reach this particular residual. The reader should refer to Section \ref{sec:data1} for details about the procedure to build these graphs.

In line with the analysis developed for the stress/anxiety and weather task data sets, we now investigate the predictive accuracy of the models based on the current body fat application. The steps considered here to compute the RSS and FSMD for different partitions of the data set with 251 observations (individual $39$ included) are exactly the same as defined in Section \ref{sec:data1}. The boxplots in Figure \ref{fig_rss2} $(a)$ indicate that the RSS obtained via bessel regression are smaller than those from the beta model. In terms of FSMD, Panel $(c)$ does not exhibit differences between bessel and beta. Panel $(b)$ shows points related to the ratio $\mbox{RSS}^{(j)}_{\tiny \mbox{bessel}}/\mbox{RSS}^{(j)}_{\tiny \mbox{beta}}$ for each partition $j = 1, \cdots, 1000$. The vast majority of the points ($79.4\%$ of them) are located below the horizontal grey line representing the equality of RSS's. It is clear from this result that $\mbox{RSS}_{\tiny \mbox{bessel}} < \mbox{RSS}_{\tiny \mbox{beta}}$ for most partitions. The visual inspection of Panel $(d)$ suggests a slightly larger ammount of FSMD ratios below the grey line. In fact, 72.1\% of the points are below the level 1 in this case, thus $\mbox{FSMD}_{\tiny \mbox{bessel}} < \mbox{FSMD}_{\tiny \mbox{beta}}$ occurs for most partitions. In summary, the analysis of Figure \ref{fig_rss2} indicates that the bessel regression can be a strong competitor to the beta model, in terms of predictive accuracy, even in an application where its goodness-of-fit is not superior.

\begin{figure}[!h]
	\centering
	\includegraphics[scale=0.30]{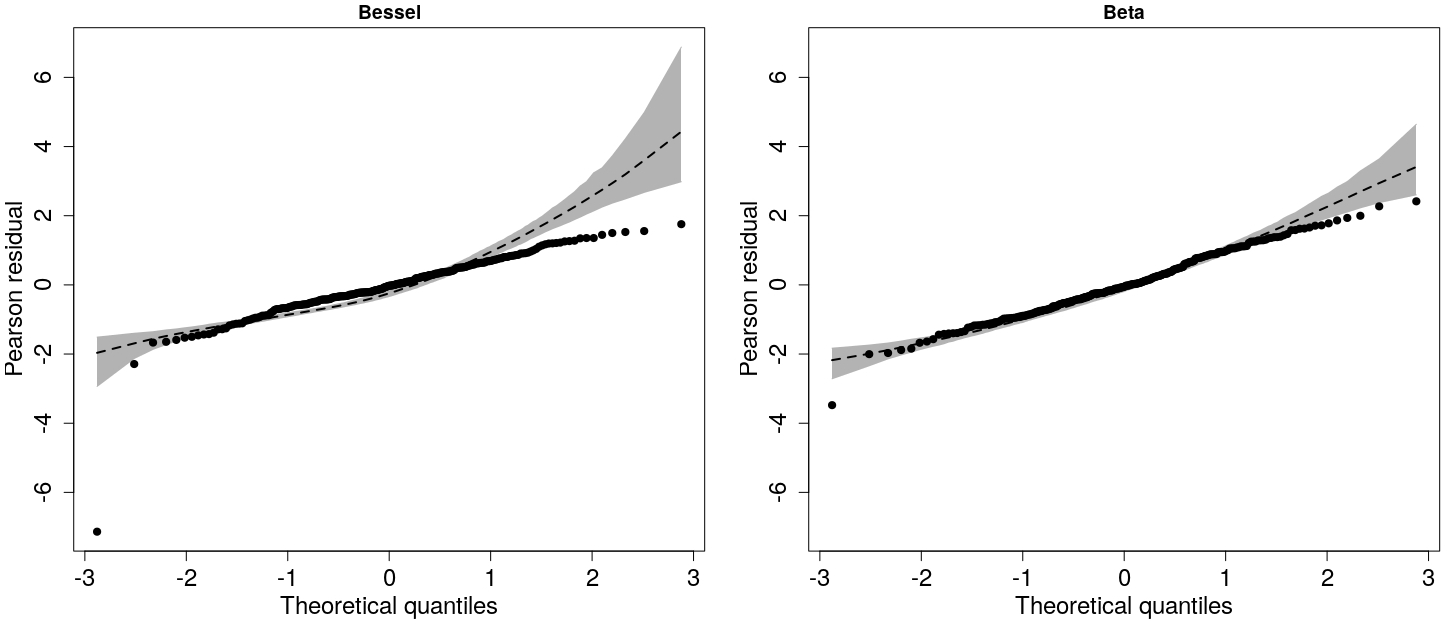}
	\caption{Pearson residuals against the theoretical quantiles from the standard normal distribution. The points are the residuals for the body fat data set. The shaded area represents a $95\%$ envelope based on $1000$ simulations from each regression. The dashed line indicates the envelope mean. Subject 39 is considered in the analysis.} \label{fig_env2}
\end{figure}

\begin{figure}[!h]
	\centering
	$$
	\begin{array}{cc}
	(a) \hspace{3.5cm} (b) & (c) \hspace{3.5cm} (d) \\	
	\includegraphics[scale=0.22]{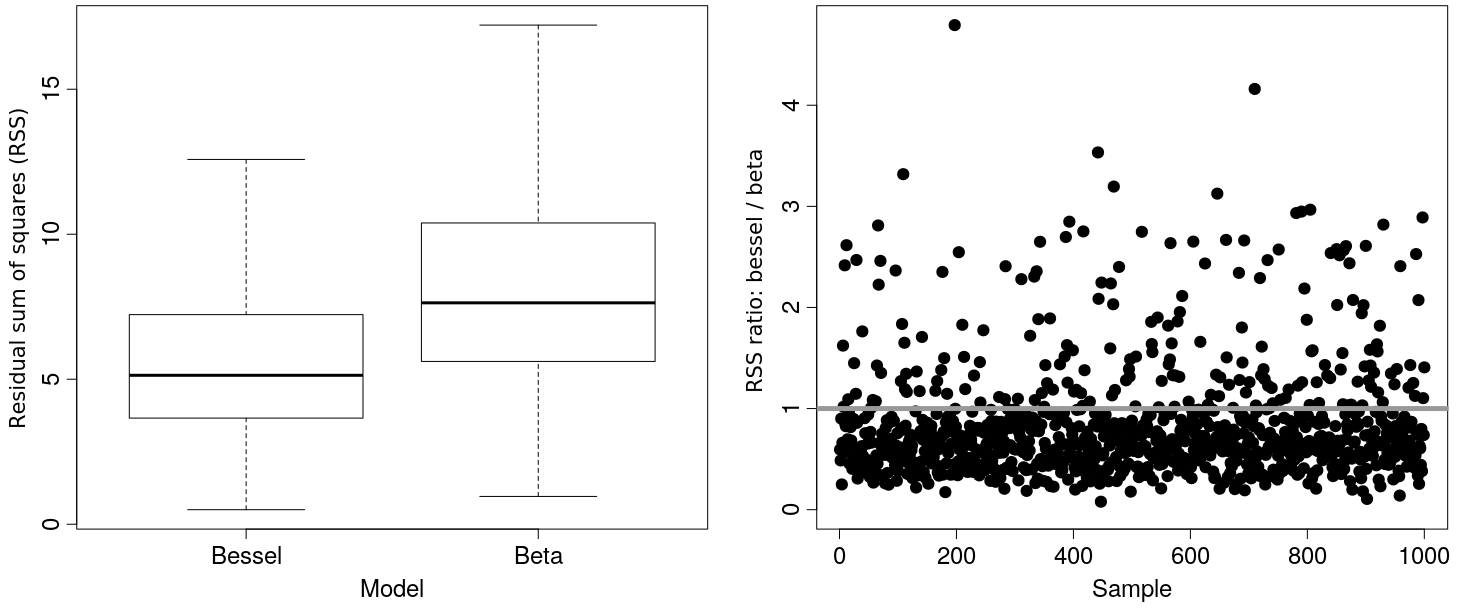} & 
	\includegraphics[scale=0.22]{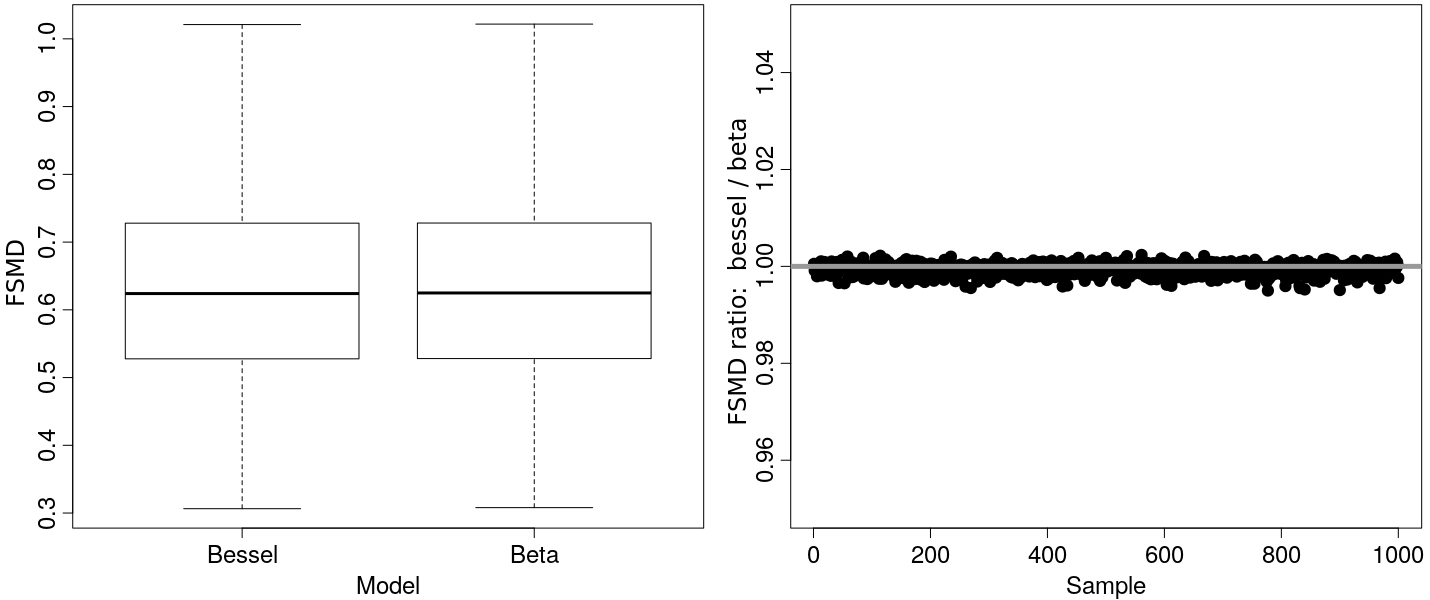} \\
	\end{array}
	$$
	\vspace{-15pt}
	\caption{Predictive accuracy of the regression models. Cross validation study based on partitions ($1000$ partitions having a training and a test set) of the body fat data set. Pearson residual and FSMD statistic are computed for $10$ randomly selected observations forming the test set. Panels $(a)$ and $(c)$ show boxplots of the RSS and the FSMD for each model, respectively. Panels $(b)$ and $(d)$ indicate the ratio of RSS's and FSMD's (bessel over beta) for each partition, respectively} \label{fig_rss2}
\end{figure}

\section{Concluding remarks}\label{sec:concluding}
In this paper we introduced the bessel regression model and showed that this is a robust alternative to the beta regression. We provided point estimation of the parameters through an EM algorithm and discussed inference. The finite-sample performance of the EM estimators were studied through Monte Carlo simulations. These simulated results indicated a good performance of the proposed algorithm in terms of inference. A discrimination test was proposed in order to select between bessel and beta regression models and its efficiency was verified through a short simulation study with synthetic data. Our proposed model showed good performances and also accuracy for prediction even under misspecification. This was illustrated in the simulated results given in Section \ref{sec:miss} and in the third empirical application with body fat data.

Other publicly available data sets passing the discrimination test in Section \ref{sec:discrimination} with indication of bessel regression are: the ``quality of education in Colombia'' study in \cite{Cepeda13} and the ``student sodium intake'' data explored in \url{http://rcompanion.org/handbook}.  

An \texttt{R} package to apply the DBB criterion proposed in this paper and to fit both bessel and beta regression models based on the EM algorithm is under development as a future supporting material for this paper. This tool will be attractive for practitioners and researches from different areas, allowing the use of our model in a joint data analysis with the beta regression. Diagnostic tools and local influence for the bessel regression model are topics to be attacked in a future paper. Other interesting points are: (i) to propose a multivariate bessel regression model for dealing with multivariate/clustered bounded data; (ii) comparison with other existing models in the literature rather than the beta; (iii) propose and study alternative link functions for the mean and precision parameters.

\section*{Appendix}

Here, we describe the elements of the information matrix in (\ref{infmatrix}). For $i = 1, \cdots, n$, let $\psi_i = E\left(W_i^{-1}|Z_i=z_i;\bs\theta\right)$ 
and $\chi_i=E\left(W_i^{-2}|Z_i=z_i;\bs\theta\right)$, where explicit expressions are obtained from Corollary \ref{Estep}. The terms forming the information matrix are obtained as follows:
\begin{eqnarray*}
	E\left(-\dfrac{\partial^2\ell_c(\bs\theta)}{\partial\kappa_j\kappa_l}\Big|{\bf Z}\right)=\sum_{i=1}^n\left\{2+\dfrac{\psi_i\phi_i^2}{z_i(1-z_i)}\left[\mu_i(1-\mu_i)-(1-2\mu_i)(z_i-\mu_i)\right]\right\}\mu_i(1-\mu_i)x_{ij}x_{il},
\end{eqnarray*}
for \ $j,l \ = \ 1, 2, \cdots,p$.
\begin{eqnarray*}
	E\left(-\dfrac{\partial^2\ell_c(\bs\theta)}{\partial\lambda_j\lambda_l}\Big|{\bf Z}\right)=\sum_{i=1}^n\left\{2\phi_i\psi_i\left(1+\dfrac{(z_i-\mu_i)^2}{z_i(1-z_i)}\right)-1\right\}\phi_i v_{ij} v_{il},\quad j,l \ = \ 1,2,\cdots,q.
\end{eqnarray*}
\begin{eqnarray*}
	E\left(-\dfrac{\partial^2\ell_c(\bs\theta)}{\partial\kappa_j\lambda_l}\Big|{\bf Z}\right)=-2\sum_{i=1}^n\phi_i^2\psi_i\mu_i(1-\mu_i)\dfrac{z_i-\mu_i}{z_i(1-z_i)} x_{ij} v_{il},\quad j = 1,\cdots,p \;\; \mbox{and} \;\; l=1,\cdots,q.
\end{eqnarray*}
\begin{eqnarray*}
	E\left(\dfrac{\partial\ell_c(\bs\theta)}{\partial\kappa_j}\dfrac{\partial\ell_c(\bs\theta)}{\partial\kappa_l}\Big|{\bf Z}\right)=\sum_{i=1}^n\Bigg\{(1-2\mu_i)^2+2\mu_i(1-\mu_i)(1-2\mu_i)\psi_i\phi_i^2\dfrac{z_i-\mu_i}{z_i(1-z_i)} + \\
	+ \chi_i\phi_i^4\mu_i^2(1-\mu_i)^2\dfrac{(z_i-\mu_i)^2}{z_i^2(1-z_i)^2}\Bigg\} x_{ij} x_{il}
	+\sum_{i\neq k}\left\{1-2\mu_i+\psi_i\phi_i^2\mu_i(1-\mu_i)\dfrac{z_i-\mu_i}{z_i(1-z_i)}\right\}\times\\
	\times \left\{1-2\mu_k+\psi_k\phi_k^2\mu_k(1-\mu_k)\dfrac{z_k-\mu_k}{z_k(1-z_k)}\right\} x_{ij}x_{kl},\quad j,l=1,2,\cdots,p.
\end{eqnarray*}
\begin{eqnarray*}
	E\left(\dfrac{\partial\ell_c(\bs\theta)}{\partial\lambda_j}\dfrac{\partial\ell_c(\bs\theta)}{\partial\lambda_l}\Big|{\bf Z}\right) \; = \; \sum_{i=1}^n\bigg\{(2+\phi_i)^2-2(2+\phi_i)\psi_i\phi_i^2\left(1+\dfrac{(z_i-\mu_i)^2}{z_i(1-z_i)}\right)+\\
	+\chi_i\phi_i^4\bigg(1+\dfrac{(z_i-\mu_i)^2}{z_i(1-z_i)}\bigg)\bigg\} v_{ij} v_{il}
	+\sum_{i\neq k}\left\{2+\phi_i-\psi_i\phi_i^2\left(1+\dfrac{(z_i-\mu_i)^2}{z_i(1-z_i)}\right)\right\}\times\\
	\times \left\{2+\phi_k-\psi_k\phi_k^2\left(1+\dfrac{(z_k-\mu_k)^2}{z_k(1-z_k)}\right)\right\}v_{ij} v_{kl},\quad j,l = 1,2,\cdots,q.
\end{eqnarray*}
Finally,
\begin{eqnarray*}
	E\left(\dfrac{\partial\ell_c(\bs\theta)}{\partial\kappa_j}\dfrac{\partial\ell_c(\bs\theta)}{\partial\lambda_l}\Big|{\bf Z}\right) \; = \; \sum_{i=1}^n\Bigg\{(1-2\mu_i)(2+\phi_i)-(1-2\mu_i)\psi_i\phi_i^2\left(1+\dfrac{(z_i-\mu_i)^2}{z_i(1-z_i)}\right)+\\
	+\mu_i(1-\mu_i)(2+\phi_i)\psi_i\phi_i^2\dfrac{z_i-\mu_i}{z_i(1-z_i)}-
	\mu_i(1-\mu_i)\chi_i\phi_i^4\dfrac{z_i-\mu_i}{z_i(1-z_i)}\left(1+\dfrac{(z_i-\mu_i)^2}{z_i(1-z_i)}\right)\Bigg\} x_{ij} v_{il}
	+\\
	+\sum_{i\neq k}\left\{1-2\mu_i+\mu_i(1-\mu_i)\psi_i\phi_i^2\dfrac{z_i-\mu_i}{z_i(1-z_i)}\right\}
	\left\{1-2\mu_k+\mu_k(1-\mu_k)\psi_k\phi_k^2\dfrac{z_k-\mu_k}{z_k(1-z_k)}\right\}x_{ij}v_{kl},
\end{eqnarray*}
for \; $j = 1,\cdots,p$ \; and \; $l = 1,\cdots,q$.
\section*{Acknowledgements}
	We would like to thank Prof. Fernando Quintana (Pontificia Universidad Cat\'olica de Chile - Santiago, Chile) for detailed discussions and suggestions on an earlier version of this paper.
 The first and second authors gratefully acknowledge the financial support from FAPEMIG (Brazil). The first and third authors acknowledge the support from CNPq (Brazil).



\end{document}